\documentclass[11pt]{amsart}
 \usepackage{etex}
\usepackage[utf8]{inputenc}
\usepackage[T1]{fontenc}
\usepackage{geometry}
\usepackage{lmodern}
\usepackage{amsmath}
\usepackage{amssymb}
\usepackage{amsthm}
\usepackage{mathrsfs}
\usepackage{graphicx}
\usepackage{pst-all}
\usepackage{marvosym}
\usepackage{epic}
\usepackage{latexsym} 
\usepackage{amsmath} 
\usepackage{epsfig}
\usepackage{amssymb} 
\usepackage{enumerate}
\usepackage{listings}
\usepackage[linesnumbered, norelsize]{algorithm2e}
\usepackage[mathscr]{euscript}
\usepackage{hyperref}

\title{Beyond representing orthology relations by trees}
\author{K.T. Huber, G. E. Scholz}
\thanks{School of Computing Sciences, 
University of East Anglia, UK}

\date{\today}

\newtheorem{theorem}{Theorem}[section]
\newtheorem{lemma}[theorem]{Lemma} 
\newtheorem{corollary}[theorem]{Corollary}
\newtheorem{proposition}[theorem]{Proposition}

\begin{document}

\maketitle

\begin{abstract}
  Reconstructing the evolutionary past of a family of genes
  is an important aspect of many genomic studies. To help with
  this, simple operations on a set of sequences called
  orthology relations may be employed. In addition to being interesting from
  a practical point of view they are also attractive from a theoretical
  perspective in that e.\,g.\,a characterization is known
  for when such a relation is representable by a certain type of
  phylogenetic tree. For an orthology relation inferred from
  real biological data it is however generally too
  much to hope for that it satisfies that
  characterization. Rather than trying to correct
  the data in some way or another which has its own drawbacks,
  as an alternative, we propose
  to represent an orthology relation $\delta$ in
  terms of a structure more general
  than a phylogenetic tree called a phylogenetic network.
  To compute such a network in the form
  of a level-1 representation for $\delta$, we introduce
  the novel {\sc Network-Popping} algorithm which has several
  attractive properties. In addition,
  we characterize orthology relations $\delta$ on some set $X$
  that have a level-1 representation in terms of eight
  natural properties for $\delta$
  as well as in terms for level-1 representations
  of orthology relations on certain subsets of $X$.
  \end{abstract}

\section{Introduction}

  Unraveling the evolutionary past of a family $\mathcal G$ of genes
  is an important aspect for many genomic studies.
  For this, it is generally assumed that the genes in $\mathcal G$ are
  orthologs, that is, have arisen from a common ancestor through
  speciation. However it is known that shared ancestry of genes can
  also arise
  via whole genome duplication (paralogs). This potentially obscures the signal
  used for reconstructing the evolutionary
  past of the genes in $\mathcal G$ in the form of a gene tree (essentially
  a rooted tree whose leaves are labelled by the elements of $\mathcal G$ --
  we present precise definitions of the main concepts used in the next
  section).
  To tackle this problem, tree-based approaches have been proposed. These
  typically work by reconciling a gene tree with
  an assumed further tree (species tree) in terms of a map that
  operates on their vertex sets. For this, certain
  evolutionary events are postulated such
  as the ones mentioned above (see e.\,g.\,\cite{N13} for a recent review as
  well as e.\,g.\,\cite{MM15} and the references therein).
  
  To overcome the problem that the resulting reconciliation
  very much depends on the quality of the
  employed trees and also that such approaches can be
  computationally demanding for larger datasets, orthology relations have been
  proposed as an alternative.
  These  operate directly on the set of
  sequences from which a gene tree is built (see e.\,g.\,\cite{Altenhoff:09}).
  In addition to having attractive practical properties, such relations
  are also interesting from a theoretical point of view due to their
  relationship with e.\,g.\,co-trees (see e.\,g.\,\cite{HHHMSW13,HW15}).
  Furthermore, a characterization is known for
  when an orthology relation can be represented in terms of a certain
  type of phylogenetic tree \cite{HHHMSW13}. 

  Due to e.\,g.\,errors or noise in an orthology relation,
  it is however in general too much to hope for that
  an orthology relation obtained from a real biological dataset
  satisfies that characterization. A natural strategy therefore might
  be to try and correct for this
  in some way. As was pointed out in \cite{LeM15} however,
  even if an underlying tree-like evolutionary scenario is assumed 
  for this many natural formalizations lead to NP-complete
  problems. 
  Furthermore,  true non-treelike evolutionary
  signal such as hybridization might be overlooked. As an alternative,
  we propose to represent orthology relations in terms of 
  phylogenetic networks. These naturally 
  generalize phylogenetic trees by permitting additional edges.
   To infer such a structure from an orthology relation $\delta$,
  we introduce the novel {\sc Network-Popping} algorithm
\begin{figure}[h]
\begin{center}
  \includegraphics[scale=0.6]{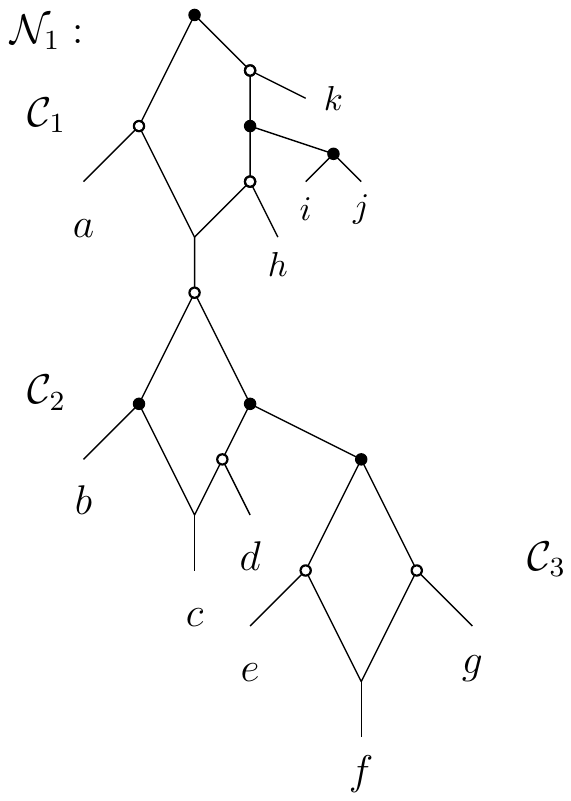}
  \includegraphics[scale=0.6]{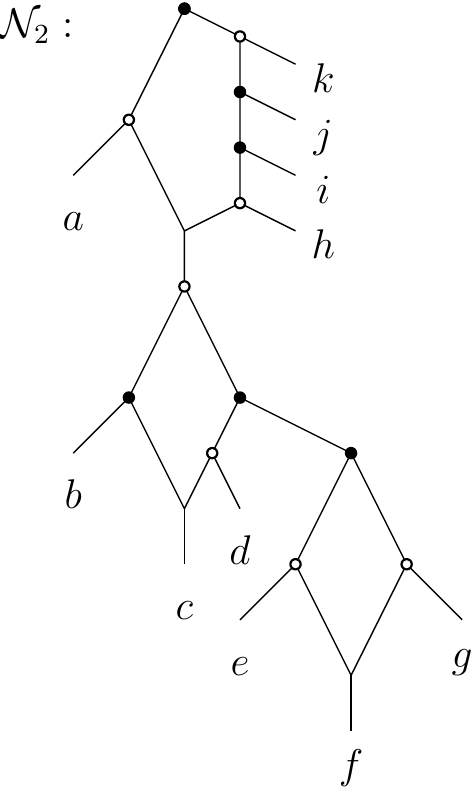}
  \includegraphics[scale=0.6]{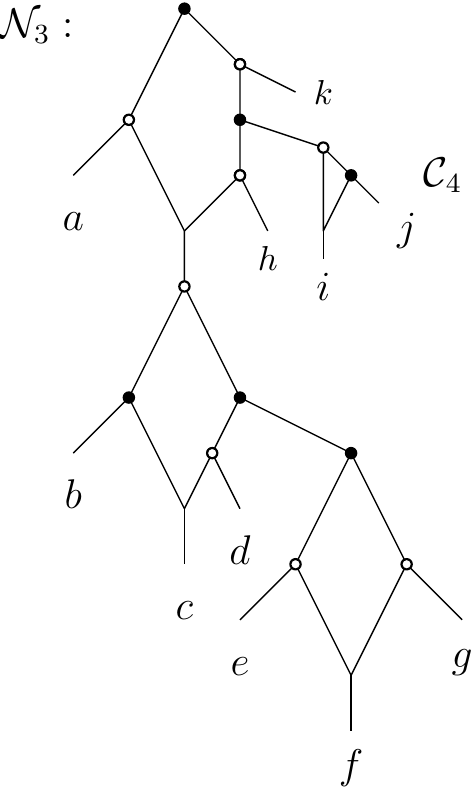}
  \caption{Three distinct level-1 representations
    for the symbolic 3-dissimilarity $\delta_{\mathcal N_2}$
    on $X=\{a,\ldots,k\}$
    induced by $\mathcal N_2$.  However, only $\mathcal N_1$ is returned by
    {\sc Network-Popping} when given $\delta_{\mathcal N_2}$.
    In all three cases the underlying phylogenetic network
      is a level-1 network.
    Furthermore, $\mathcal N_2$,
is not semi-discriminating but weakly labelled whereas  $\mathcal N_3$
is semi-discriminating but not weakly labelled
-- See text for details.
\label{exnet}}
\end{center}
\end{figure}
which returns a level-1 representation of $\delta$
in the form of a structurally very simple phylogenetic network called
a level-1 network (see e.\,g.\,Fig.~\ref{exnet} for examples of
such representations where the interior vertices labelled in terms
of $\bullet$ and $\circ$ represent two distinct evolutionary events
such as speciation and whole genome duplication and the unlabelled
interior vertices represent hybridization events).

Bearing in mind
the point made in \cite[Chapter 12]{F03}, that
$k$-estimates, $k\geq 3$, are potentially more accurate
than mere distances as they capture more information,
we formalize an orthology relation in terms of a
symbolic 3-dissimilarity rather than a symbolic 2-dissimilarity (i.e.
a distance), as was the case in \cite{HHHMSW13}.
From a technical point of view this also allows us to overcome the
problem that using a symbolic 2-dissimilarity in a network context
can be problematic. An example illustrating this
is furnished by the three level-1 representations
depicted in Fig.~\ref{3ns3d} which all represent the same
2-dissimilarity induced by taking the lowest common ancestor between
pairs of leaves.

\begin{figure}[h]
\begin{center}
\includegraphics[scale=0.6]{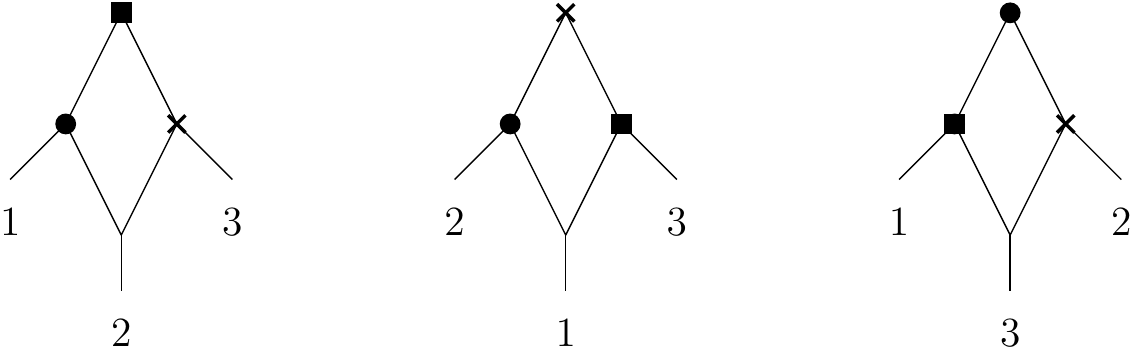}
\label{3ns3d}
\caption{Three distinct level-1 representations of the 
  2-dissimilarity
  $\delta : {\{1,2,3\}\choose 2} \to M=\{\bullet, \times, \blacksquare\}$
  defined by taking lowest common ancestors of pairs of leaves.
}
\end{center}
\end{figure}

As we shall see, algorithm {\sc Network-Popping} is
guaranteed to find, in polynomial time,
a level-1 representation of a symbolic 3-dissimilarity
if such a representation exists. For this, it
relies on the three further algorithms below which
we also introduce. It works by first finding
for a symbolic 3-dissimilarity $\delta$ on
$X$ all pairs of subsets of $X$ that support a cycle using algorithm
\textsc{Find-Cycles}. Subsequent to this, it
employs algorithm \textsc{Build-Cycles} to
construct from each such pair $(H,R')$ a structurally
very simple level-1 representation for the symbolic 3-dissimilarity
induced on $H\cup R'$.
Combined with algorithm \textsc{Vertex-Growing}
which constructs a symbolic discriminating representation
for a symbolic 2-dissimilarity,  
{\sc Network-Popping} then recursively grows the level-1
representation for $\delta$
by repeatedly
applying algorithms \textsc{Build-Cycles} and \textsc{Vertex-Growing}
in concert.
For the convenience of the reader, we illustrate all four algorithms
by means of the level-1 representations depicted in Fig.~\ref{exnet}.
As part of our analysis of algorithm {\sc Network-Popping},
we characterize level-1 representable
symbolic 3-dissimilarities $\delta$ on $X$ in terms of
eight natural properties (P1) -- (P8) enjoyed by $\delta$.
(Theorem~\ref{theo:char}).
Furthermore,
we characterize such dissimilarities in terms of level-1 representable
symbolic 3-dissimilarities
on subsets of $X$ of size $|X|-1$ (Theorem~\ref{iffsubsets}).
Within a Divide-and-Conquer framework the resulting speed-up of
algorithm {\sc Network-Popping} might allow it
to also be applicable to large datasets.

 The paper is organized as follows. In the next section, we
present basic definitions and results. Subsequent to this,
we introduce in Section~\ref{sec:delta-trinets} the crucial concept of
a $\delta$-trinet associated to a symbolic 3-dissimilarity
and state Property~(P1).
In Section~\ref{sec:find-cycles}, we present algorithm {\sc Find-Cycles}
as well as Properties~(P2) and (P3).
In Section~\ref{sec:build-cycles}, we introduce and analyse
algorithm {\sc Build-Cycles}.
Furthermore, we state Properties (P4) -- (P6).
In Section~\ref{sec:network-popping}, we present algorithms
{\sc Vertex-Growing}
and {\sc Network-Popping}. As suggested by
the example in Fig.~\ref{exnet}, algorithm {\sc Network-Popping}
need not return the level-1 representation of a symbolic 3-dissimilarity
that induced it. Employing a further algorithm
called {\sc Transform}, we address in Section~\ref{sec:uniqueness} the
associated uniqueness question (Corollary~\ref{cor:uniqueness}).
As part of this we establish
Theorem~\ref{theo:char}
which includes stating Properties (P7) and (P8). In
Section~\ref{sec:characterizing-level-1}, we
establish Theorem~\ref{iffsubsets}. We conclude with
Section~\ref{sec:conclusion}
where we present research directions that might be worth 
pursuing.

\section{Basic definitions and results}
\label{sec:prelim}

In this section, we collect relevant basic terminology and results
concerning phylogenetic networks and 
symbolic $2$- and $3$-dissimilarities.
From now on and unless stated otherwise, $X$ denotes a finite set of 
size $n \geq 3$, $M$ denotes a finite set of symbols of size at least two
and $\odot$ denotes a symbol not already contained in $M$.
Also, all directed/undirected graphs have no loops or multiple 
directed/undirected edges.
 
\subsection{Directed acyclic graphs}
Suppose $G$ is a rooted directed acyclic graph (DAG), that is, a
DAG with a unique vertex with indegree zero. We call that
vertex the {\em root} of $G$, denoted by $\rho_G$. Also, we call the graph
$U(G)$ obtained from $G$ by ignoring the directions of its edges
the {\em underlying graph} of $G$. By abuse of terminology, we call
an induced
subgraph $H$ of $G$ a {\em cycle} of $G$ if the induced subgraph
$U(H)$ of $U(G)$ is a cycle of $U(G)$. We call a vertex $v$ of $G$
an {\em interior vertex} of $G$ if $v$ is not
a leaf of $G$ where we say that a vertex $v$ is a {\em leaf} if
the indegree of $v$
is one and its outdegree is zero. We denote the set of interior vertices
of $G$ by $V(G)_{int}$ and the set of leaves of $G$
by $L(G)$. We call a vertex $v$ of $G$ a {\em tree vertex} if the
indegree of $v$ is at most one and its 
outdegree is at least two, and
a {\em hybrid vertex} of $G$ if the indegree of $v$ is two
and its outdegree is not zero.
The set of interior vertices of $G$ that are
not hybrid vertices of $G$ is denoted by $V(G)_{int}^-$. 
 We say that $N$ is {\em binary} if, with the
 exception of $\rho_N$,  the indegree and outdegree of each of its
 interior vertices add up to three.
Finally, we say that two DAG's $N$ and $N'$ with leaf set $X$ 
are {\em isomorphic} if there exists a bijection
from $V(N)$ to $V(N')$ that extends to a (directed) graph 
isomorphism between $N$ and $N'$ which is the identity on $X$.

\subsection{Phylogenetic networks and last common ancestors}
A {\em (rooted) phylogenetic network $N$ (on $X$)}
is a rooted DAG that does not contain a vertex that has
indegree and outdegree one and $L(N)=X$.
In the special case that a phylogenetic network $N$  is
such that each of its interior vertices
belongs to at most one cycle we call $N$ a
{\em a level-1 (phylogenetic) network (on $X$)}.
Note that a phylogenetic network may contain cycles
of length three and that a phylogenetic network that
does not contain a cycle is called
a {\em phylogenetic tree $T$ (on $X$)}.

For the following, let $N$ denote a level-1 network on $X$.
For $Y\subseteq X$ with $|Y|\geq 3$, we denote by 
$N|_Y$ the subDAG of $N$ induced by $Y$ (suppressing any resulting
vertex that have indegree and outdegree one). Clearly, $N|_Y$
is a phylogenetic network on $Y$. 
 
 
 Suppose $v$ is a non-leaf vertex of $N$. We say that a further
 vertex $w\in V(N)$ is
 {\em below} $v$ if there is a directed path from $v$ to $w$ and call
 the set of leaves of $N$ below $v$ the \emph{offspring set} of $v$,
 denoted by $\mathcal F(v)$. Note that $\mathcal F(v)$
 is closely related to the hardwired cluster of $N$
 induced by $v$ (see e.g. \cite{HRS10}).
 For a leaf $x\in \mathcal F(v)$, we refer
 to $v$ as an {\em ancestor} of $x$. In case
 $N$ is a phylogenetic tree, we define
 the \emph{lowest common ancestor} $lca_N(x,y)$ of two distinct  leaves
 $x,y\in L(N)$ to be the (necessarily unique) ancestor $v\in V(N)$
 such that $\{x,y\} \subseteq \mathcal F(v)$ and 
 $\{x,y\} \nsubseteq \mathcal F(v')$ holds for all children
 $v'\in V(N)$ of $v$. More
generally, for
$Y \subseteq X$ with $2 \leq |Y| \leq |X|$, we denote by $lca_N(Y)$
 the unique vertex $v$ of $N$ such that $Y \subseteq \mathcal F(v)$,
 and $Y \nsubseteq \mathcal F(v')$ holds for all children
 $v'\in V(N)$ of  $v$. Note that in case the tree $N$ we are
 referring too is clear from the context, we shall write $lca(Y)$
 rather than $lca_N(Y)$.

 It is easy to see that the notion of a lowest common ancestor
 is not well-defined for phylogenetic networks in general.
 However the situation changes in case the network in question
 is a level-1 network, as
 the following central result shows. Since its proof is
 straight-forward, we omit it.

\begin{lemma}\label{lcrefaunique}
  Let $N$ be a level-1 network on $X$ and assume that $Y\subseteq X$ such
  that $|Y|\geq 2$.
 Then there exists a unique interior vertex $v_Y\in V(N)$ such that 
$Y\subseteq \mathcal F(v_Y)$ but $Y\not\subseteq \mathcal F(v')$, 
for all children $v'\in V(N)$ of $v_Y$. Furthermore, there exists two 
distinct elements $x,y \in Y$ such that $v_Y=v_{\{x,y\}}$.
\end{lemma}

Continuing with the terminology of Lemma~\ref{lcrefaunique},
we shall refer to $v_Y$ as the
{\em lowest common ancestor of $Y$ in $N$}, denoted by $lca_N(Y)$.
As in the case of a phylogenetic tree, we shall write
$lca(Y)$ rather than $lca_N(Y)$ if the network $N$ we
are referring to is clear from the context.

\subsection{Symbolic dissimilarities and labelled level-1 networks}
Suppose $k\in\{2,3\}$. We denote by ${X \choose k}$ 
the set of subsets of $X$ of size $k$, and by
${X \choose \leq k}$ the set of nonempty subsets of $X$ of
size at most $k$. We
call a map $\delta : {X \choose \leq k} \to M^{\odot}:=M \cup \{\odot\}$
a {\em symbolic $k$-dissimilarity on $X$ with values
in $M^{\odot}$} if, for all
$A \in {X \choose \leq k}$,
we have that $\delta(A)=\odot$ if and only if $|A|=1$.
To improve clarity of exposition, we shall refer to $\delta$ as a
{\em symbolic $3$-dissimilarity on $X$} if the set $M$ is of no relevance
to the discussion. Moreover, for $Y=\{x_1, \ldots, x_l\}$, $l\geq 2$,
we shall write $\delta(x_1, \ldots, x_l)$ rather than
$\delta(Y)$ where the order of the
elements $x_i, 1 \leq i \leq l $, is of no
relevance to the discussion.

A {\em labelled (phylogenetic) network $\mathcal N=(N,t) $  (on $X$)}
 is a pair consisting of a phylogenetic network
 $N$ on $X$ and a labelling map
 $t:V(N)_{int}^-\to  M$. If $N$ is a level-1 network
 then $\mathcal N$ is called a {\em labelled level-1 network}. 
 To improve clarity of exposition
we shall always use calligraphic font to denote a
labelled phylogenetic network.

Suppose $\mathcal N=(N,t) $ is a labelled
level-1 network on $X$ such that its vertices in $V(N)_{int}^-$
are labelled in terms of $M$. Then we denote by
$\delta_{\mathcal N}: {X \choose \leq 3} \to M^{\odot}$
the symbolic 3-dissimilarity
on $X$ induced by $\mathcal N$ given by $\delta_{\mathcal N}(Y)=t(lca(Y))$
if $|Y|\not=1$, and $\delta_{\mathcal N}(Y)= \odot$ otherwise.
For $\mathcal N'=(N',t')$ a further
labelled level-1 network on $X$, we say that
$\mathcal N$ and $\mathcal N'$ 
are {\em isomorphic} if $N$ and $N'$ are isomorphic and
$\delta_{\mathcal N}=\delta_{\mathcal N'}$.

Conversely, suppose $\delta$ is a symbolic 3-dissimilarity 
on $X$. In view of Lemma~\ref{lcrefaunique}, we call 
a labelled level-1 network $\mathcal N=(N, t)$ on $X$ a
\emph{level-1 representation} of $\delta$
if $\delta=\delta_{\mathcal N}$. For ease of
terminology, we shall sometimes say that
$\delta$ is {\em level-1 representable} if the
the labelled network we are referring too is of no relevance
to the discussion.
We call a level-1 representation of $\delta$
{\em semi-discriminating}
if $N$ does not contain a directed edge $(u,v)$
such that $t(u)=t(v)$ except for when there exists
a cycle $C$ of $N$ with $|V(C)\cap \{u,v\}|=1$. 
For example, all three labelled level-1 networks depicted in 
Fig.~\ref{exnet} are level-1 representations of $\delta_{\mathcal N_2}$
where $\mathcal N_2$ is the labelled level-1 network depicted in
Fig.~\ref{exnet}(ii).
Furthermore, the representations of 
$\delta_{\mathcal N_2}$ presented in Fig.~\ref{exnet}(i) and (iii),
respectively,
are semi-discriminating  whereas the one depicted in
Fig.~\ref{exnet}(ii) is not.

\begin{figure}[h]
\begin{center}
\includegraphics[scale=0.5]{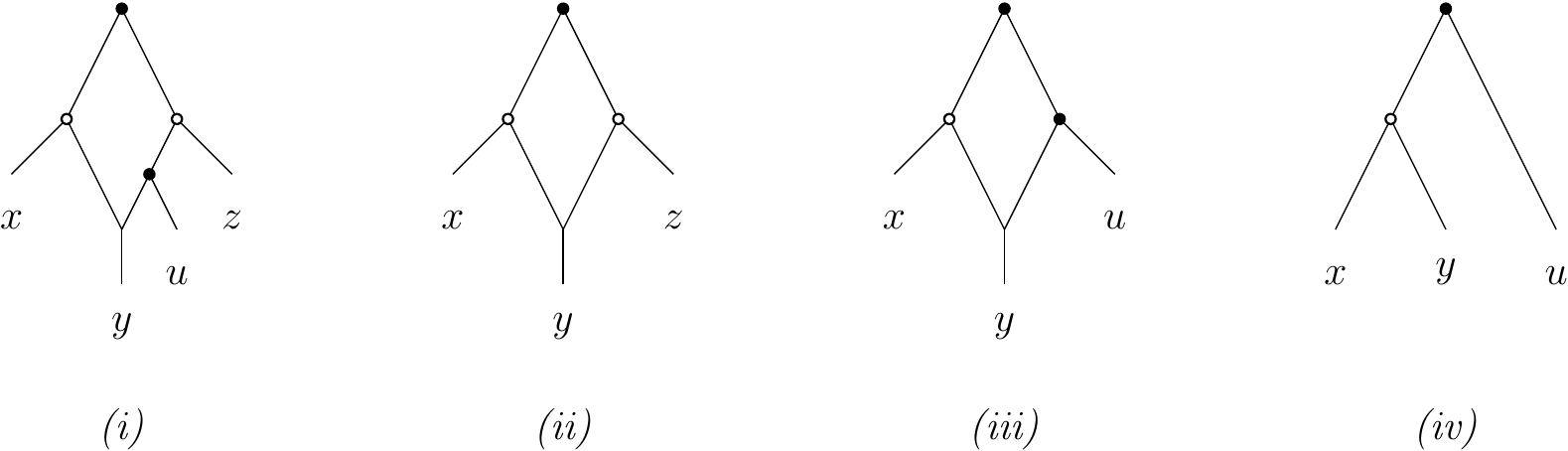}
\caption{\label{dispnu}(i) A labelled level-1 network $\mathcal N$ on 
  $X=\{x,y,z,u\}$. (ii) and (iv) Semi-discriminating level-1 representations
  of $\delta_{\mathcal N}$ restricted to  $\{x,y,z\}$ and
  $Y=\{u,x,y\}$, respectively. (iii) A level-1 representation
  of $\delta_{\mathcal N}|_Y$ in the form of a labelled trinet 
  that is
  is not a $\delta_{\mathcal N}$-trinet.
}
\end{center}
\end{figure}

Note that in case $N$ is a phylogenetic tree on $X$ the 
definition of a semi-discriminating level-1 representation for $\delta$ 
reduces to that of a discriminating symbolic representation 
for the restriction $\delta_2=\delta|_{{X\choose \leq 2}}$ of $\delta$
to ${X\choose 2}$ (see \cite{BD98} and also \cite{HHHMSW13,SS03} for more
on such representations). 
Using the concept of a {\em symbolic ultrametric}, that is, a
symbolic 2-dissimilarity
$\delta :  {X \choose \leq 2} \to M^{\odot}$
for which, in addition, the following two properties are satisfied
\begin{enumerate}
\item[(U1)] $|\{\delta(x,y),\delta(x,z),\delta(y,z)\}| \leq 2$ for all
$x, y, z \in X$;
\item[(U2)] there exists no four elements $x, y, z, u \in X$ such that
  \[\delta(x,y)=\delta(y,z)=\delta(z,u)\neq\delta(z,x)=\delta(x,u)=
  \delta(u,y);\]
\end{enumerate}
such representations were characterized by the authors  
of \cite{BD98} as follows.

\begin{theorem}\cite[Theorem 7.6.1]{BD98} \label{bdtr}
  Suppose $\delta : {X \choose \leq 2} \to M^ {\odot}$ is a
  2-dissimilarity on $X$. Then there exists a discriminating
  symbolic representation of
  $\delta$ if and only if $\delta$ is a symbolic ultrametric.
\end{theorem}



Clearly, it is too much to hope for that 
any symbolic 3-dissimilarity
$\delta$ has a level-1 representation.
The question therefore becomes: Which symbolic 3-dissimilarities
have such a representation? A first partial answer is
provided by Theorem~\ref{bdtr} and Lemma~\ref{lcrefaunique}
for not $\delta$ but its
restriction $\delta_2$. More precisely, $\delta$
has a discriminating symbolic representation if and only if 
$\delta_2$ is a symbolic ultrametric and,
for all $x,y,z \in X $ distinct, $\delta(x, y, z)$
is the (unique) element appearing at least twice
in the multiset $\{\delta_2(x,y), \delta_2(x,z),\delta_2(y,z)\}$.

\section{$\delta$-triplets, $\delta$-tricycles, and $\delta$-forks}
\label{sec:delta-trinets}
To make a first inroad into the aforementioned
question, we next investigate structurally very simple level-1
representations of symbolic 3-dissimilarities. As we shall see, these
will turn out to
be of fundamental importance for our algorithm {\sc Network-Popping}
(see Section~\ref{sec:network-popping}) as well as for our analysis of its
properties. In the context of this, it is important
to note that although
{\em triplets} (i.\,e.\,binary phylogenetic trees on 3 leaves)
are well-known to uniquely determine (up to isomorphism)
phylogenetic trees this does not hold for
level-1 networks in general \cite{GH12}. To overcome this
problem, {\em trinets},
that is, phylogenetic networks on three leaves
were introduced in \cite{HM13}. For the convenience of the
reader, we depict in Fig.~\ref{12t} all 12 trinets $\tau_1,\ldots,\tau_{12}$ 
on $X=\{x,y,z\}$ from  \cite{HM13}
that are also level-1 networks in our sense.
In the same paper
it was observed that even the slightly more general
1-nested networks are uniquely determined by their induced trinet sets
(see also \cite{HIMS15} for more on constructing level-1 networks from
trinets, and \cite{vIM14} for an extension of this result to
other classes of phylogenetic networks).
\begin{figure}[h]
\begin{center}
\includegraphics[scale=0.5]{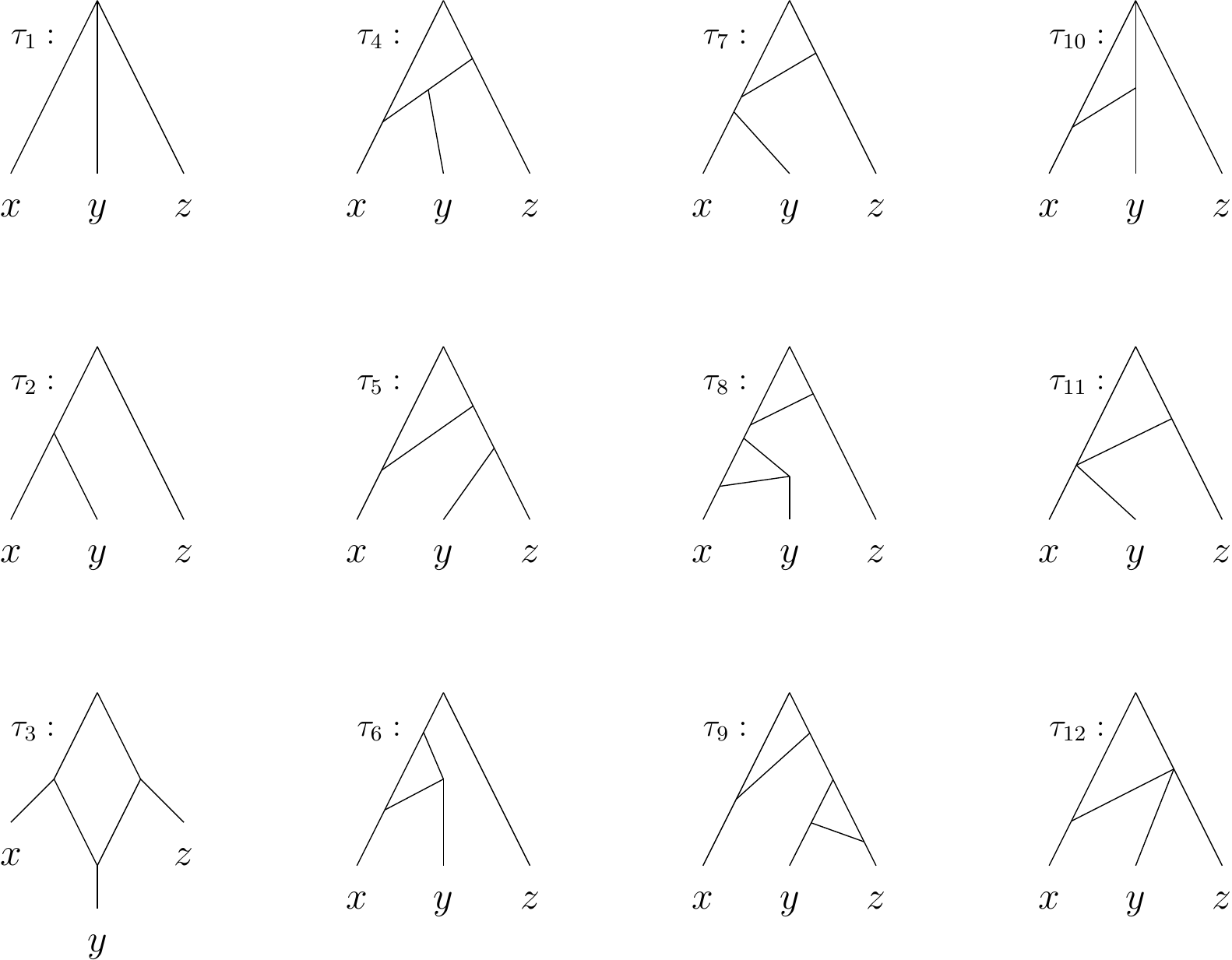}
\caption{The twelve trinets in the from of
  level-1 networks. The two omitted trinets from \cite{HM13} 
  are not level-1 networks in our sense.
\label{12t}}
\end{center}
\end{figure}

Perhaps not surprisingly, trinets on their own
are not strong enough to uniquely determine
labelled level-1 networks in the sense that
any two level-1 representations of a 
symbolic 3-dissimilarity must be isomorphic.
 To see this, suppose $|X|=3$ and
consider the symbolic 3-dissimilarity
$\delta:{X\choose \leq 3}\to \{A,B,\odot\}$ that  maps $X$ and every
2-subset of $X$ to $A$. Then the labelled network
$(\tau_1,t)$ where $t$ maps every vertex in $V(\tau_1)_{int}^-$  
to $A$ is a semi-discriminating level-1 representation
of $\delta$ and so is the labelled network
$(\tau_4,t')$, where 
every vertex in $V(\tau_4)_{int}^-$ is mapped  to $A$
by $t'$.  Note that similar arguments
may also be applied to the level-1 representations involving
the trinet $\tau_4$ to $\tau_{12}$ depicted in Fig.~\ref{12t}.
We therefore evoke parsimony and focus for the remainder of this paper
on the trinets $\tau_1$, $\tau_2$ and $\tau_3$. We shall refer to 
them as {\em fork} on $X=\{x,y,z\}$, {\em triplet} $z|xy$,
and {\em tricycle} $y||xz$, respectively.

The next result (Lemma~\ref{3tri}) relates forks, triplets and tricycles 
with symbolic 3-dissimilarities. To state it, we say that
a symbolic 3-dissimilarity $\delta$ satisfies the {\em Helly-type Property} 
if, for any three elements $x,y,z\in X$, we have
$\delta(x,y,z) \in \{\delta(x,y),\delta(x,z),\delta(y,z)\}$. Note that
we will sometimes also refer to the Helly-type property as Property
(P1).

\begin{lemma}\label{3tri}
Suppose $\delta$ is a symbolic 3-dissimilarity on a set $X=\{x,y,z\}$ 
taking values in $M$. Then there exists a 
level-1 representation 
$\mathcal N$ of $\delta$ if and only if 
 $\delta$ satisfies the Helly-type Property.
In that case $\mathcal N$ can be (uniquely) chosen
to be semi-discriminating and, (up to permutation  of the leaves 
of the underlying level-1 network $N$) 
$N$ is isomorphic to one 
of the trinets $\tau_1$, $\tau_2$ and $\tau_3$
depicted in Fig.~\ref{12t}.
\end{lemma}
\begin{proof}
Suppose first that $\mathcal N=(N,t)$ is a level-1 representation of
$\delta$. Then, in view of Lemma~\ref{lcrefaunique},
 $\delta(x,y,z) \in \{\delta(x,y),\delta(x,z),\delta(y,z)\}$ must hold.

Conversely, suppose that  
$\delta(x,y,z) \in E:= \{\delta(x,y),\delta(x,z),\delta(y,z)\}$ 
holds for all elements 
$x,y,z\in X$ distinct. By analyzing the size of $E$ it is
straight-forward to show that one of the situations indicated
in the rightmost column
of Table~\ref{tabtri} 
\begin{table}[h]
\begin{center}
\begin{tabular}{|c|c||c|}
\hline
$|\{\delta(x,y),\delta(x,z),\delta(y,z)\}|$ & $\delta(x,y,z)=...$ & 
$N$ \\
\hline
\hline
1 & $\delta(x,y)=\delta(x,z)=\delta(y,z)$ & fork \\
\hline
3 & $\delta(y,z)$ & $x || yz$ \\
\hline
2 & $\delta(y,z) \neq \delta(x,y)=\delta(x,z)$ &  $x || yz$ \\
\hline
2 & $\delta(x,y)=\delta(x,z)$ &$x|yz$ \\
\hline
\end{tabular}
\caption{\label{tabtri}For  
$\delta:{X\choose \leq 3}\to M^{ \odot}$ 
  a symbolic 3-dissimilarity we list all labelled trinets on
  $X=\{x,y,z\}$ in terms of the size of $E$.}
\end{center} 
\end{table}
must apply. With defining a labelling map $t:V(N)_{int}^-\to M^{\odot}$
in the obvious way using the second column of that table, it follows
that $\mathcal N$ is a level-1 representation for $\delta$. 
\end{proof}

Armed with Lemma~\ref{3tri}, we make the following central definition.
Suppose that $|Y|=3$, that $\delta$
is a symbolic 3-dissimilarity on $Y$, and that $\mathcal N=(N,t)$
is a semi-discriminating level-1 representation of $\delta$. Then we call
$\mathcal N$ a
{\em $\delta$-fork} if $N$ is a fork on $Y$,
a {\em $\delta$-triplet} if $N$ is a triplet on $Y$, and
a {\em $\delta$-tricycle} if $N$ is a tricycle on $Y$,
 For ease of
 terminology, we will collectively refer to
 all three of them as a {\em $\delta$-trinet}.
Note that as the example of the labelled trinet depicted in
Fig.~\ref{dispnu}(iii) shows, there exist trinets that are
not $\delta$-trinets. By abuse of terminology, we shall
refer for a symbolic 3-dissimilarity $\delta$ on $X$
and any 3-subset $Y\subseteq X$ to a
$\delta|_Y$-trinet as a $\delta$-trinet.

\section{Recognizing cycles: The algorithm \textsc{Find-Cycles}}
\label{sec:find-cycles}
In this section, we introduce and analyze algorithm {\sc Find-Cycles}
(see Algorithm~\ref{afc} for a pseudo-code version).
Its purpose
is to recognize cycles in a level-1 representation of a symbolic
3-dissimilarity $\delta$ if such a representation exists. As we shall
see, this algorithm relies on Property (P1) and a certain
graph $\mathcal C(\delta)$ that can be canonically associated to
$\delta$.
Along the way, we also establish two further crucial properties enjoyed 
by a level-1 representable symbolic 3-dissimilarity.

We start with
introducing further terminology. Suppose $N$ is a level-1
network and $C$ is a cycle of $N$. Then we denote  
by $r(C)$ the unique vertex in $C$ for
which both children are also contained in $C$ and 
call it the {\em root} of  $C$. In addition, we call
the hybrid vertex of $N$
contained in $C$ the {\em hybrid} of $C$
and denote it by  $h(C)$. Furthermore,
we denote set of all elements of $X$ below  by $r(C)$
by  $R(C)$ 
and the set of all elements of $X$ below $h(C)$
by  $H(C)$.
Clearly, $H(C) \subsetneq R(C)$.
 Moreover, for any leaf $x\in R(C) - H(C)$, 
 we denote by $v_{C}(x)$ the last ancestor of $x$ in $C$.
 Note that $v_{C}(x)$ is the parent of $x$ if and only
 if $x$ is incident with a vertex in $C$.
 Last-but-not-least, we 
 call the vertex sets of the two edge-disjoint directed paths
 from $r(C)$ to $h(C)$ the 
\emph{sides} of $C$. Denoting these two
paths by $P_1$ and $P_2$, respectively, we say that two leaves $x$ and $y$ in 
$R(C)-H(C)$ $\emph{lie on the same side}$ of
 $C$ if the vertices $v_{C}(x)$ and $v_{\mathcal C}(y)$
are both interior vertices of $P_1$ or $P_2$ , 
and that they $\emph{lie on different sides}$ if they are not.
For example, for $C$ the underlying cycle
of the cycle $\mathcal C_2$ indicated in the
labelled network $\mathcal N_1$ pictured in Fig.~\ref{exnet}(i),
we have $R(C)=\{b,\ldots,g\}$  and $H(C)=\{c\}$.
Furthermore, the sides of $C$ are $\{ r(C),v_{C}(b),h(C) \}$
and $\{r(C), v_{C}(d), v_C(e),  h(C)\}$ and $d,\ldots,g$ lie
on one side of $C$ whereas $b$ and $d$ lie on different sides of $C$.
%
%
%

Suggested by Property (U2), the following property is of interest
to us where $\delta$ denotes again a 
symbolic 3-dissimilarity on $X$: 
\begin{enumerate}
\item[{\em (P2)}] For all
$x,y,z,u \in X$ distinct for which
$\delta(x,y)=\delta(y,z)=\delta(z,u) \neq \delta(z,x)=\delta(x,u)=
\delta(u,y)$ holds there exists exactly one subset 
$Y \subseteq \{x,y,z,u\}$ of size $3$ such that a tricycle on $Y$
underlies a level-1 representation of
$\delta|_Y$.
\end{enumerate}

As a first result, we obtain
%


\begin{lemma}\label{gencond}
  Suppose $\delta$ is a level-1 representable
  symbolic 3-dissimilarity on $X$.
  Then $\delta$ satisfies the Helly-type Property as well as Property (P2).
\end{lemma}

\begin{proof}
Note first that Property (P1) is a straight-forward consequence of
Lemma~\ref{lcrefaunique}.
 
To see that Property (P2) holds, note first that since
$\delta$ is level-1 representable there
exists a labelled level-1 network
$(N,t)$ such that $\delta(Y)=t(lca(Y))$,
for all subsets $Y\subseteq X$ of size $2$ or $3$.
Suppose $x,y,z,u\in X$ distinct are such that 
$\delta(x,y)=\delta(y,z)=\delta(z,u) \neq \delta(z,x)=\delta(x,u)=
\delta(u,y)$. To see that there exists some $Y\subseteq Z:=\{x,y,z,u\}$ 
for which $(N|_Y,t|_Y)$ is a $\delta$-tricycle,
 assume for contradiction that there exists no 
such set $Y$.  
By Theorem~\ref{bdtr}, $N$ cannot be a phylogenetic tree on $X$ and, so,
$N$ must contain at least one cycle $C$. Without loss of generality,
we may assume that $x\in H(C)$, and $y$ lies on one of the two sides 
of $C$. By assumption $\delta(y,z) \neq \delta(x,z)$ 
and so either $z$ and $y$ lie on opposite sides of $C$,
or $z$ and $y$ lie on the same side of $C$ and $v_C(y)$ lies on the
directed path from $r(C)$ to $v_C(z)$. As can be easily checked,
either one of these two cases
yields a contradiction since then $\delta(z,u) \neq \delta(x,u)=\delta(y,u)$
cannot hold for $u$, as required.

To see that there can exist at most one such tricycle 
on $Z$, assume for contradiction that there exist tow tricycles  
$\tau$ and $\tau'$ with $L(\tau)\cup L(\tau')\subseteq Z$.
Then $ |L(\tau)\cap L(\tau')|=2$. Choose $x,y\in L(\tau)\cap L(\tau')$.
Note that the assumption on the elements of $Z$ implies that
$x$ or $y$ must be below the hybrid vertex of one of $\tau$
and $\tau'$ but not the other. Without loss of generality we may
assume that $y$ is
below the hybrid vertex of $\tau$ but not below the hybrid vertex of $\tau'$.
Then $y$ must lie on a side of the unique cycle $C'$ of $\tau'$.
But this is impossible since the unique cycle of $\tau$  and
$C'$ are induced by the same cycle of $N$.
\end{proof}


We remark in passing that the proof of uniqueness in the proof of 
Lemma~\ref{gencond} combined with 
 the structure of a level-1 network, readily implies
 the following result.

\begin{lemma}\label{lmtri}
Suppose that $\delta$ is 
a symbolic 3-dissimilarity on $X$ that is level-1 representable by a labelled
network $(N,t)$
and that $x,y,z\in X$ are three distinct elements such that $x || yz$ is a 
 $\delta$-tricycle. Let $C$ denote 
the unique cycle in $N$ such that 
$x \in H(C)$ and $y,z \in R(C)-H(C)$, and let 
$x'\in X$.  If  $x' || yz$ is a 
 $\delta$-tricycle then $x' \in H(C)$ and if $x || x'z$
is a $\delta$-tricycle then 
$x' \in R(C)$ and $x'$ and $y$ lie on the same side of $C$.
\end{lemma} 

To better understand the structure of a symbolic 3-dissimilarity
$\delta$, we 
next associate to $\delta$ a graph $\mathcal C(\delta)$ defined as
follows. The vertices
of $\mathcal C(\delta)$ are the
$\delta$-tricycles and any two $\delta$-tricycles $\tau$ and $\tau'$
are joined by an edge if $|L(\tau)\cap L(\tau')|=2$. For
example, consider 
the symbolic 3-dissimilarity $\delta_{\mathcal N_1}$ 
induced by the labelled level-1 network $\mathcal N_1$
pictured in Fig.~\ref{exnet}(i). Then  the graph presented
in Fig.~\ref{cdelta} is $\mathcal C(\delta_{\mathcal N_1})$. 
\begin{figure}[h]
\begin{center}
\includegraphics[scale=0.6]{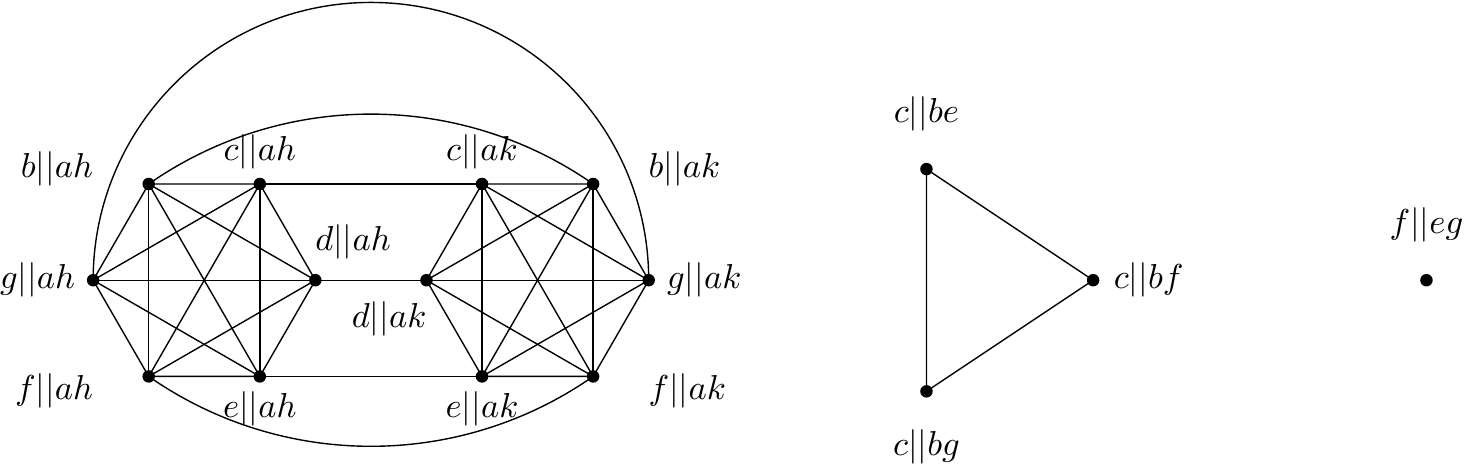}
\caption{\label{cdelta}The graph $\mathcal C(\delta_{\mathcal N_1})$, 
  where $\mathcal N_1$ is the labelled level-1 network depicted in
  Fig.~\ref{exnet}(i).}
\end{center}
\end{figure}

The example in Fig.~\ref{cdelta} suggests the following property
for a symbolic 3-dissimilarity $\delta$
to be level-1 representable:
\begin{enumerate}
\item[{\em (P3)}] If $\tau$ and $\tau'$ are $\delta$-tricycles 
contained in the same 
connected component of $\mathcal C(\delta)$, then 
$$\delta(L(\tau))=\delta(L(\tau')).$$
\end{enumerate}

We collect first results concerning Property (P3) in the next proposition.

\begin{proposition}\label{cocy}
Suppose  $\delta:{X\choose \leq 3} \to M^{\odot}$
is a symbolic 3-dissimilarity. If $\delta$ is 
level-1 representable or  $|M|=2$ holds
then Property (P3) must hold.
In particular, if $\mathcal N$ is a level-1 
representation for $\delta$ then there exists a canonical
injective map from the set of connected 
components of $\mathcal C(\delta)$ to the set of cycles of
the level-1 network underlying $\mathcal N$.
\end{proposition}
\begin{proof}
  Suppose first  that $\delta$ is level-1 representable. Let
  $\mathcal N=(N,t)$
denote a level-1 representation of $\delta$.
Then $\delta=\delta_{\mathcal N}$. Since
$\delta_{\mathcal N}(x,y,z)=t(r(C))$ holds
for all cycles $C$ of $N$, 
and any $x \in H(C)$ and any $y,z\in R(C)$ that lie
on different sides of $C$, 
Property (P3) follows.

Suppose next that $|M|=2$. 
It suffices to show that Property (P3) holds for any two
adjacent vertices of $\mathcal C(\delta)$.
Suppose $\tau$ and $\tau'$ are two such vertices
and that $x,y,z\in X$ are such that
$\tau= x||yz$. Then there exists some $u\in X$ such that
either $\tau'=u||yz$ or $\tau'=x||ru$ where $r\in\{y,z\}$.
Without loss of generality we may assume that $r=y$.
In view of the Table~\ref{tabtri}, we clearly have
$\delta(x,y)\not=\delta(x,y,z)=\delta(y,z)$. Since, in addition, 
$\delta(u,y,z)=\delta(y,z)$ holds in the former case it follows that  
$\delta(L(\tau))=\delta(L(\tau'))$. In the latter case,
we obtain $\delta(x,y,u) \neq \delta(x,y)$ and thus,
$\delta(L(\tau))=\delta(L(\tau'))$ follows in this case too as
$|M|=2$.

The claimed injective map
is a straight-forward consequence of Lemma~\ref{lmtri}.
\end{proof}

Algorithm \textsc{Find-Cycles} exploits the
injection mentioned in Proposition~\ref{cocy} by
interpreting for a symbolic 3-dissimilarity $\delta$ 
a connected component $C$ of $\mathcal C(\delta)$ in terms of 
two sets $H_C$ and $R_C'$.  Note that if
$C'$ is a cycle in the level-1 network underlying a level-1 representation
of $\delta$ (if such a representation exists!),
the sets $H(C')$ and $H_C$ coincide and $R_C' \subseteq R(C')$
holds. 

\begin{algorithm}
\label{afc}
\caption{\textsc{Find-Cycles} -- Property (P1) is checked in Line 1.}
\SetKw{KwSet}{set}
\KwIn{A symbolic 3-dissimilarity $\delta$ on $X$.}
\KwOut{A number $m \geq 1$ and pairs of subsets $(H_i, R'_i)$ of $X$,
  $1\leq i\leq m$,
  or the statement ``$\delta$ is not level-1 representable''.}
\BlankLine
\If{$\delta$ satisfies Property \emph{(P1)}}{
Build the graph $\mathcal C(\delta)$\;
Denote by $m$ the number of connected components of $\mathcal C(\delta)$\;
\For{$i \in \{1, \ldots, m\}$}{
Let $K_i$ denote a connected component of $\mathcal C(\delta)$\;
\KwSet {$H_i=\{x \in X : \text{ there exist } y, z \in X  
\text{ such that }x|| yz \text{ is a vertex of } K_i\}$}\;
\KwSet{$R'_i=H_i \cup \{y \in X : \text{ there exist } x, z \in X
  \text{ such that } x|| yz \text{ is a vertex of  } K_i\}$}\;
}
\Return{$m, (H_1,R'_1), \ldots, (H_m,R'_m)$}\;
}
\Else{
\Return{$\delta$ is not level-1 representable}\;
}
\end{algorithm}

For example, for the symbolic 3-dissimilarity $\delta_{\mathcal N_1}$ 
induced by the labelled network $\mathcal N_1$ depicted in
Fig.~\ref{exnet}(i), algorithm \textsc{Find-Cycles} returns
the three pairs $(b\ldots g,a\ldots, k)$, $(c,bcefg)$ 
and $(f,efg)$ where we write $x_1\ldots x_{|A|}$
for a set $A=\{x_1,\ldots, x_{|A|}\}$.

\section{Constructing cycles: The algorithm \textsc{Build-Cycles}
\label{sec:build-cycles}}

We next turn our attention toward reconstructing a structurally
very simple level-1 representation of a symbolic 3-dissimilarity
(should such a representation exist).
For this, we use algorithm \textsc{Build-Cycles} which
takes as input a symbolic 3-dissimilarity $\delta$
and a pair returned by \textsc{Find-Cycles} when given $\delta$.

To state \textsc{Build-Cycles}, we require further terminology.
Suppose $N$ is a level-1 network. Then 
we say that $N$ is {\em partially resolved}
 if all vertices in a cycle of
 $N$ have degree three. Note that partially-resolved level-1 networks may have
 interior vertices not contained in a cycle that have degree three or more.
 Thus such networks need not be binary.
If, in addition to being partially resolved, $N$ is such that it
 contains a unique cycle $C$
such that every non-leaf vertex of $N$ is a vertex of 
$C$ then we call $N$ {\em simple}.

Algorithm \textsc{Build-Cycle} (see Algorithm~\ref{abc}
for a pseudo-code version) relies on a further graph
called the TopDown graph
associated to a symbolic 3-dissimilarity $\delta$.
For $(H,R')$ a pair returned by algorithm \textsc{Find-Cycle}
when given $\delta$
and $x\in H$ and $S\subseteq R'$,  that graph
essentially orders the vertices of $S$. Thus, for each connected
component $K$ of $\mathcal C(\delta)$,
\textsc{Build-Cycle} computes a level-1 representation of $\delta$
corresponding to
$K$ (should  such a
representation exist).

We start with presenting a central observation concerning labelled
level-1 networks.

\begin{lemma}\label{tripcyc}
Suppose $\mathcal N=(N,t)$ is  a labelled level-1 network,
 and $C$ is a cycle of $N$. Suppose also that 
$x,y,z \in X$ are three elements such that
 $x\in H(C)$, $y,z \in R(C)-H(C)$ 
and $t(v_{C}(z))=t(r(C)) \neq t(v_{C}(y))$.
 Then, $v_{C}(z)$ lies on the directed path from $v_{C}(y)$ 
to $h(C)$ if and only if $y|xz$ is a $\delta_{\mathcal N}$-triplet. 
\end{lemma}

\begin{proof}
  Put $\delta=\delta_{\mathcal N}$.
  Suppose first that $v_{C}(z)$ lies on the directed path from
$v_{C}(y)$ to $h(C)$. Then
 $lca(x,y,z)=lca(x,y)=lca(y,z)=v_{C}(y)$ and
 $lca(x,z)=v_{C}(z)$.
Hence, $\delta(x,y,z)=\delta(x,y)=
\delta(y,z)=t(v_{C}(y))\not=t(v_{C}(z))=
\delta(x,z)$. By  Table~\ref{tabtri}, 
$y|xz$ is a $\delta$-triplet.

Conversely,
suppose that $y|xz$ is a $\delta$-triplet.
Then, by Table~\ref{tabtri}, we have
$\delta(x,y,z)=\delta(x,y)=
\delta(y,z) \neq \delta(x,z)$. 
Since $\delta(x,y)=t(v_{C}(y))$ 
and $\delta(x,z)=t(v_C(z))$, it follows that
$\delta(x,y,z)=t(v_{C}(y)) 
\neq t(v_{C}(z))$. But then $y$ and $z$ must lie on the same
side of  $C$ as otherwise 
$\delta(y,z)=t(r(C))$ follows  
which is impossible by assumption on $x$, $y$ and $z$.
Thus, either $v_{C}(y)$ must lie on a directed path $P$ from
$v_{C}(z)$ to $h(C)$ or $v_{C}(z)$ must lie on a directed path $P'$ from
$v_{C}(y)$ to $h(C)$.
However $v_{C}(y)$ cannot be a vertex on $P$
as otherwise $lca(y,z)=v_{C}(z)$ holds and, so,
$\delta(y,z)=\delta(x,z)$ follows, which is 
impossible. Thus  $v_{C}(z)$ must be a vertex on $P'$.
\end{proof}

With $\mathcal N$ and $C$ as in from Lemma~\ref{tripcyc},
it follows from  Lemma~\ref{lmtri}, that whenever
algorithm \textsc{Find-Cycles}
is given $\delta_{\mathcal N}$ as input, it
returns a pair $(H,R')$ such that $H=H(C)$ and
$R'=H(C) \cup \{y \in R(C) : t(v_C(y)) \neq t(r(C))\}$. Moreover
giving $(H,R')$ and $\delta_{\mathcal N}$ as input to algorithm 
\textsc{Build-Cycle},
Lemma~\ref{tripcyc} implies that \textsc{Build-Cycle} finds all elements 
$z\in R(C)-R'$ for which there exists some 
$y \in R'$ such that $v_{C}(z)$ lies on the path 
from $v_{C}(y)$ to $h(C)$. However it should be noted that
if $z\in R(C)-H(C)$ is such that $t(v)=t(r(C))=t(v_{C}(z))$
holds for all vertices $v$ on the path from $r(C)$ to $v_{C}(z)$
then the information captured by $\delta_{\mathcal N}$ for
$x$, $y$, and $z$ is in general not sufficient 
to decide if $z$ and $y$ 
lie on the same side of $C$ or not. In fact, it is easy to see that,
in general, $z\in R(C)$ need not even hold.

We now turn out attention to the aforementioned TopDown graph associated to
a symbolic 3-dissimilarity  $\delta$ on $X$ which is defined as follows.
Suppose that 
$S \subsetneq X$, and that $x \in X-S$. Then the vertex set of
the {\em TopDown graph $TD(S,x)$} is $S$ and two 
elements $u,v\in S$ distinct are joined by a 
direct edge $(u,v)$ if $u|vx$ is a $\delta$-triplet.
%
  \begin{figure}[h]
\begin{center}
  \includegraphics[scale=0.6]{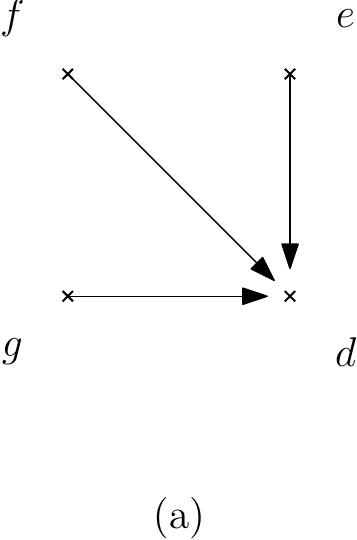},\,\,\,\,\,\,\,\,\,\,\,\,\,\,\,\, 
   \includegraphics[scale=0.6]{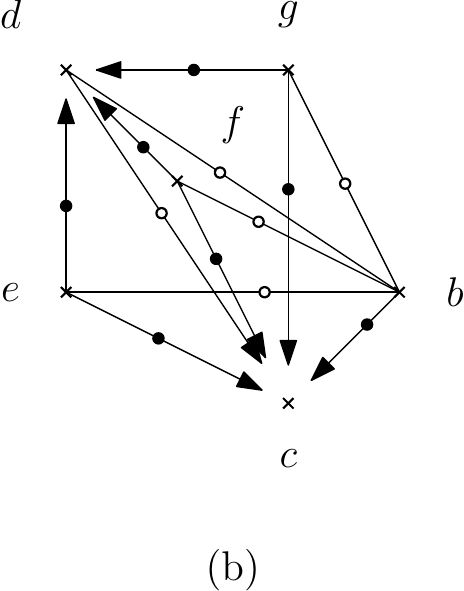}
   \caption{\label{fig:td-graph}
   For $\delta_{\mathcal N_1}$ the symbolic $3$-dissimilarity induced by the
   labelled network $\mathcal N_1$ pictured in Fig.~\ref{exnet}(i),
   we depict in (a) the graph $TD(\{d,e,f,g\}, c)$ and in (b)
   the graph $CL(\{c\}, \{b\}, \{d,e,f,g\})$. In both graphs, the vertices
   are indicated by ``$\times$''. -- See text for details.}
     \end{center}
\end{figure}


\begin{algorithm}
\label{abc}
\caption{\textsc{Build-Cycle} -- The set $R'$ is the set $H\cup S_y\cup S_z$,
  Property (P4) is checked in Lines 6, 10, and 20, 
and  Properties (P3), (P6), (P7) and (P8) are checked in 
  Lines 6, 13, 10 and 20, respectively.-- See text for details.}
\SetKw{KwSet}{set}
\KwIn{A symbolic 3-dissimilarity $\delta$ on $X$ that satisfies
  Property (P1) and a pair $(H,R')$ returned by algorithm
  \textsc{Find-Cycle} when given $\delta$.}
\KwOut{Either a labelled simple level-1 network $(C, t)$ on a
  partition of a subset $X'$ of $X$ such that $R'\subseteq X'$ and
  $H(K)=H$ holds for the unique cycle $K$ of $C$, or
  the statement ``$\delta$ is not level-1 representable''.}
\BlankLine
\KwSet {rep=0}\;
Choose a $\delta$-tricycle $x || yz$, where $x \in H$ and $y, z \in R'-H$\;
\KwSet {$S'_y=\{u \in R' : x || uz \text{ is a }
  \delta \text{-tricycle}\}$}\;
\KwSet {$S'_z=\{u \in R' : x || yu \text{ is }
  \delta \text{-tricycle}\}$}\;
Initialize $C$ as a graph with three vertices respectively labelled
by $r(C)$, $h(C)$ and $H$, and the edge $(h(C),H)$\;
\If {for all $x' \in H$, $y' \in S'_y$ and $z' \in S'_z$, $x'||y'z'$ is
  a $\delta$-tricycle and $\delta(x,y,z)=\delta(x',y',z')$}{
\KwSet {$t(r(C))=\delta(x,y,z)$}\;
\KwSet {$S_y=S'_y \cup \{u \in X-R' : \text{ there exists }
  u' \in S'_y \text{ such that } u'|ux \text{ is a } \delta
  \text{-triplet}\}$}\;
\KwSet {$S_z=S'_z \cup \{u \in X-R' : \text{ there exists }
  u' \in S'_z \text{ such that } u'|ux \text{ is a } \delta
  \text{-triplet}\}$}\;
\If {for all $u_1 \in S_y, u_2 \in S_z, \delta(u_1,u_2)=t(r(C))$}{
\For{$i \in \{y,z\}$}{
\KwSet {$v_l=r(C)$}\;
\If {$TD(S_i,x')=TD(S_i,x'')$ for all $x', x'' \in H$ and $TD(S_i,x)$
  does not contain a directed cycle}{
\KwSet {$G=TD(S_i,x)$}\;
\KwSet {rep=rep+1}\;
\While {$V(G) \neq \emptyset$}{
Add a new child $v$ to $v_l$\;
\KwSet {$\mathcal F(v)=\{u \in S_i : u \text{ has indegree }
  0 \text{ in } G\}$}\;
Delete from $G$ all vertices in $\mathcal F(v)$\;
\If {for all $u, u' \in \mathcal F(v)$, $x', x'' \in H \cup V(G)$,
  $\delta(u,x')=\delta(u',x'')$}{
Choose some $u \in \mathcal F(v)$\;
\KwSet $t(v)=\delta(x,u)$\;
Add the leaf $\mathcal F(v)$ as a child of $v$\;
\KwSet $v_l=v$\;
}
\Else{
Remove all vertices from $G$\;
\KwSet {rep=rep-1}\;
}
}
Add the edge $(v_l,h(C))$\;
}
}
}
}
\If{rep=2}{
\Return {$C$}\;
}
\Else{
\Return {$\delta$ is not level-1 representable}\;
}

\end{algorithm}

Rather than continuing with our analysis of
algorithm {\sc Build-Cycle} we break for the moment and 
illustrate it by means of an example. For this
we return again to the symbolic 3-dissimilarity $\delta_{\mathcal N_1}$
on $X=\{a, \ldots,k\}$
induced by the labelled level-1 network $\mathcal N_1$
depicted in Fig.~\ref{exnet}(i). Suppose 
$(c,bcefg)$ is a pair returned by algorithm \textsc{Find-Cycle}
and $c||be$ is the $\delta$-tricycle chosen in line 2 of \textsc{Build-Cycle}. 
Then  $H=\{c\}$, $S'_b=\{b\}$ and $S'_e=\{e,f,g\}$ (lines 3 and 4),
and $S_b=\{b\}$ and $S_e=\{d,e,f,g\}$ (lines 8 and 9). The
graph $TD(S_e,c)$ is depicted in Fig.~\ref{fig:td-graph}(a). It implies
that 
for the cycle $C$ associated to the pair $(c,bcefg)$
in a level-1 representation of $\delta_{\mathcal N_1}$,
we must have $v_C(e)=v_C(f)=v_C(g)$ and that one of
the two sides of $C$ is $\{e,d,f,g\}$. Since 
$|S_b|=1$, the  other side of $C$ is $\{b\}$
(lines 11 to 33).

Continuing with our analysis of algorithm {\sc Build-Cycle},
we remark that the fact that the TopDown graph $TD(S_e,c)$
in the previous example is non-empty is not a coincidence.
In fact, it is easy to see that the graph $G$ defined
in line 14 of {\sc Build-Cycle}
is non-empty whenever $\delta$ is level-1 representable.
Thus, the DAG $C$ returned by
algorithm {\sc Build-Cycle}
cannot contain multi-arcs. Note however 
that there might be tricycles induced by $C$ of the form 
$x||uz$ with $u\in R'- S_y'$ as, for example,
$\delta(x,z)=\delta(x,y)=\delta(z,y)=\delta(x,u)$ might hold and thus 
$x||uz$ is not a $\delta$-tricycle. Note that similar reasoning also applies
to $S_z'$ and the extensions of $S_y'$ and $S_z'$ to $S_y$ and $S_z$
defined in lines 8 and 9, respectively. Also note that
the sets $S_y$ and $S_z$ 
are dependent on the choice of the
$\delta$-tricycle in line 2. However,
line 6 ensures that the labelled simple level-1 network returned by
algorithm \textsc{Build-Cycle} is independent of the choice of that
$\delta$-tricycle.

To establish Proposition~\ref{prop:terminate} which ensures
that algorithm {\sc Build-Cycle} terminates,
we next associate to a directed graph $G$ a new graph $P(G)$ 
by successively removing vertices of indegree zero and their
incident edges until no such vertices remain. As a first
almost trivial observation concerning that graph we have the
following straight-forward result whose proof we again  omit.

\begin{lemma}\label{proppg}
Let $G$ be a directed graph. Then $P(G)$ is nonempty 
if and only if $G$ contains a directed cycle.
\end{lemma}

Given as input to algorithm \textsc{Build-Cycle}
a symbolic 3-dissimilarity $\delta$ that satisfies Property (P1) and
a pair $(H,R')$ returned by algorithm \textsc{Build-Cycle} for $\delta$
we have

\begin{proposition}\label{prop:terminate}
  Algorithm \textsc{Build-Cycle} terminates.
\end{proposition}
\begin{proof}
  As is easy to check the only reason for algorithm \textsc{Build-Cycle}
  not to terminate is the while loop initiated in its line 16. For $i=1,2$,
  this while loop works by
  successively removing vertices of indegree 0 (and their incident edges)
  from the graph $TD(S_i,x)$,
  and terminates if the resulting graph, i.\,e.\,$P(TD(S_i,x))$, is empty.
  Since line 13 ensures that this
  loop is entered if and only if $TD(S_i,x)$ does not contain a directed cycle,
  Lemma~\ref{proppg} implies that \textsc{Build-Cycle} terminates.
  \end{proof}

It is straight-forward to see that when given a 
level-1 representable symbolic 3-dissimilarity $\delta$ 
such that the underlying level-1 network is in fact a simple level-1
network the labelled network returned by algorithm
{\sc Build-Cycle} satisfies the following three additional properties
(where we use the notations introduced in algorithm {\sc Build-Cycle}).
\begin{enumerate}
\item[\emph{(P4)}] For $i=y,z$, we have $S_i'=\{u\in S_i: \delta(u,x)\not=
  \delta(y,z)\}$ and  
  $S_y \cap S_z = S_y \cap H = S_z \cap H = \emptyset$.
\item[\emph{(P5)}] For all $u,v \in R:=H\cup S_y\cup S_z$
  and all $w \in X-R$, we have 
$\delta(u,w)=\delta(v,w)$.
\item[\emph{(P6)}] For all $u,u' \in H$ and $i\in\{y,z\}$, the graphs 
$TD(S_i,u)$ and $TD(S_i,u')$ are 
isomorphic and do not contain a directed cycle.
\end{enumerate}

Since the quantities on which these properties are
based also exist for general symbolic 3-dissimilarities 
we next study Properties (P4) - (P6) for such dissimilarities. As a
first consequence of Property (P4) combined with Properties (P1) and (P2),
we obtain a sufficient
condition
under which the TopDown graph $TD(S_i,x)$ considered in
algorithm {\sc Build-Cycle} does not contain 
a directed cycle (lines 13). For convenience, we employ
again the notation
used in Algorithm~\ref{abc}.

\begin{proposition}\label{cytdx}
Suppose that $\delta:{X\choose \leq 3} \to M^{\odot}$
is a symbolic 3-dissimilarity that satisfies Properties (P1), 
(P2) and (P4), that
$(H,R')$ is a pair returned by algorithm {\sc Find-Cycles} when
given $\delta$, and that
$x$, $y$ and $z$ are as specified as in line 2 of algorithm 
{\sc Build-Cycle}. Then the following hold for $i=y,z$.\\
\noindent (i) If  $TD(S_i,x)$ contains a directed cycle then 
it contains a directed cycle of size 3.\\
(ii) $TD(S_i,x)$ does not contain a directed cycle
of length 3 whenever $|M|=2$ holds.
\end{proposition}

\begin{proof}
  (i) 
By symmetry, it suffices to show the proposition for $i=y$.
Suppose $TD(S_y,x)$ contains a directed cycle. Over
all such cycles in $TD(S_y,x)$, choose a directed cycle
$C$ of minimal length. If $|V(C)|=3$, then 
the statement clearly holds.

Suppose for contradiction for the remainder that $|V(C)|\geq 4$.  Suppose 
$a, b, c, d\in V(C)$ are such that $(a,b)$, $(b,c)$, $(c,d)$ are three directed
edges in $C$. We next distinguish 
between the cases that $|V(C)|\geq 5$ and that $|V(C)|= 4$.

Suppose $|V(C)|\geq 5$. Then since $a,c\in S_y$, Lemma~\ref{lmtri}
combined with the minimality of $C$ implies that
we either have a $\delta$-fork on $\{a,c,x\}$  or the $\delta$-triplet
$ac|x$. Hence, $\delta(x,a)=\delta(x,c)$ holds in either case.
Note that similar arguments also imply that
$\delta(x,b)=\delta(x,d)$. Since $|V(C)|\geq 5$, 
the directed edges $(a,d)$ and $(d,a)$ cannot be contained in
$TD(S_y,x)$ and, using again
similar arguments as before, $\delta(x,a)=\delta(x,d)$ must hold.
In combination, we obtain $\delta(x,a)=\delta(x,b)$
which is impossible in view of $(a,b)$
being an edge in  $TD(S_y,x)$ and thus $\delta(x,a)\not=\delta(x,b)$.

Suppose $|V(C)|= 4$. By the minimality of $C$, neither $(b,d)$ 
$(d,b)$, $(a,c)$ nor $(c,a)$ can be a directed edge in $TD(S_y,x)$.
Using similar arguments as in the previous case, it follows that
$\delta(x,b)=\delta(x,d)$ and $\delta(x,a)=\delta(x,c)$.
Combined with the facts that $(a,b)$, $(b,c)$, $(c,d)$ are directed
edges in $C$ and that $(d,a)$ must also be an edge in $C$ as
$|V(C)|=4$, it follows that with $A:=\delta(c,d)$ and $B:=\delta(b,c)$
we have 
\begin{eqnarray}\label{eq:P2}
A=\delta(x,c)=\delta(x,a)=\delta(a,b)\not=\delta(x,b)=\delta(x,d)=\delta(d,a)=
\delta(b,c)=B.
\end{eqnarray}
Note that, $\delta(a,c)\in\{A,B\}$ must also hold as 
otherwise $|\{\delta(a,c),\delta(a,b), \delta(b,c)\}|=3$ and so, in view of 
Table~\ref{tabtri},
$\delta|_{\{a,b,c\}}$ would be level-1 representable by a
$\delta$-tricycle on $\{a,b,c\}$.
But then $H\cap \{a,b,c\}\not=\emptyset$ which is impossible 
in view of Property (P4). Similarly, one can show that 
$\delta(b,d)\in\{A,B\}$. By combining a
case analysis as indicated in Table~\ref{tabtri} with Equation~\ref{eq:P2},
  it is straight-forward to 
see that each of the four detailed combinations of $\delta(a,c)$ and
$\delta(b,d)$ in that table
yields a contradiction in view of Property (P2).

(ii) By symmetry, it suffices to assume $i=y$. Let $|M|=2$ and 
assume for contradiction that $TD(S_y,x)$ contains a directed cycle $C$ of
size 3. Let $s$, $u$, $v$ denote the 3 vertices of $C$ such that
$(s,u)$, $(u,v)$ and $(v,s)$ are the three directed edges of $C$.
Then $\delta(u,x)\not= \delta(s,x)\not=\delta(v,s)= \delta(v,x)
\not=\delta(u,v)=\delta(u,x)$ must hold. Since $|M|=2$, this is impossible.
\end{proof}

\section{Constructing level-1 representations
  from symbolic 3-dissimilarities: The algorithm \textsc{Network-Popping}
\label{sec:network-popping}}

In this section, we present algorithm
\textsc{Network-Popping} which allows us to decide if
a symbolic 3-dissimilarity is level-1 representable 
or not. If it is, then \textsc{Network-Popping} is guaranteed to 
find a level-1 representation in polynomial time.

\textsc{Network-Popping} takes as input a symbolic
3-dissimilarity $\delta$ on $X$ 
and employs a top-down approach to recursively
construct a semi-discriminating level-1 representation for $\delta$
(if such a representation exists).
 For $l$ a leaf whose label set is of size at least
two and constructed in one of the previous steps it essentially 
works by either replacing $l$ with a labelled simple
level-1 network or a labelled phylogenetic tree. To compute those
networks algorithms \textsc{Find-Cycle} and \textsc{Build-Cycle} 
are used and to construct such trees algorithm \textsc{Vertex-Growing}
is employed.  At the heart of the latter
lie Proposition~\ref{f3tsu} and
algorithm \textsc{Bottom-Up}
introduced in \cite{HHHMSW13}. The latter takes 
as input a symbolic 2-dissimilarity $\delta$
satisfying Properties (U1) and (U2), 
and builds the unique discriminating symbolic representation 
$\mathcal T$ for $\delta$ (if it exists). 

To be able to state algorithm \textsc{Vertex-Growing}, we 
require again further terminology.
Following e.\,g.\,\cite{SS03}, we call a collection $\mathcal H$
of non-empty subsets of $X$ a \emph{hierarchy on $X$} if
$A\cap B\in \{A,B,\emptyset\}$ holds for any two
sets $A,B\in\mathcal H$. 
The proof of the following result is straight-forward 
and thus omitted.

\begin{lemma}\label{intcy}
Let $N$ be a level-1 network with cycles 
$C_1, C_2, \ldots, C_k$, $k\geq 1$. Then, 
$\mathcal H_N=\{R(C_1), R(C_2), \ldots, R(C_k)\}$
is a hierarchy on $X$.
\end{lemma}

Suppose $\mathcal A$ is a set of non-empty subsets of $X$. Then we
define a relation $\sim_{(X,\mathcal A)}$ on $X$ by putting
$x \sim_{(X,\mathcal A)} y$ if there exists some  $A\in \mathcal A$ 
such that $x,y\in A$, for all $x,y\in X$. Note first that $ \sim_{(X,\mathcal A)}$
is clearly an equivalence relation whenever $\mathcal A$ is a
hierarchy. In addition, suppose that $\mathcal A$ is such that 
the partition $X'$ of $X$ induced by $ \sim_{(X,\mathcal A)}$
has size two or more. If $\delta:{X\choose \leq 3} \to M^{\odot}$
is a symbolic 3-dissimilarity
such that for any two sets $Y,Y'\in X'$ we have
$\delta(x,y)=\delta(x',y')$ for all $x,x'\in Y$ and $y,y'\in Y'$,
then we associate to $\delta$ the map $\hat{\delta}$ given by
\[\begin{array}{r c l}
\hat{\delta}: {X' \choose \leq 2}  &\to& M^{\odot}\\
\{Y_1, Y_2\} &\mapsto&
\left\{
\begin{array}{l l}
\odot &\text{ if } Y_1=Y_2,\\
\delta(y_1, y_2), \text{ where } y_1 \in Y_1, y_2 \in Y_2 &
\text{ otherwise}.
\end{array}
\right.
\end{array}\]
Note that $\hat{\delta}$
is clearly well-defined and a symbolic 2-dissimilarity on $X'$.
Associating to a level-1 representation $\mathcal N=(N,t)$ of $\delta$
the set $\mathcal R:=\{R(C)\,:\, C\mbox{ is a cycle of } N\}$,
we have the following result as an immediate consequence.

\begin{proposition}\label{f3tsu}
Suppose $\mathcal N$ is a labelled level-1 network on $X$ and 
$X'$  is the partition of $X$ induced by the
relation $\sim_{(X,\mathcal R)}$
on $X$. If $|X'|\geq 2$
then $\hat{\delta_{\mathcal N}}$
is well defined and satisfies Properties 
(U1) and (U2). In particular, $\hat{\delta_{\mathcal N}}$
is a symbolic ultrametric on $X'$.
\end{proposition}

\begin{proof}
  Put $\mathcal N=(N,t)$ and 
  $\delta'=\hat{\delta_{\mathcal N}}$. Note first
  that for all $x,y\in X$,   
Lemma~\ref{intcy} implies that there exists 
 some $R \in\mathcal R$ such that $x,y\in R$ if 
and only if there exists 
$R' \in \mathcal R':=\{R \in \mathcal R : 
R\text{ is set-inclusion maximal in } \mathcal R\}$
such that $x,y\in R'$.
Let $T_N$ denote the tree obtained from $N$ by first collapsing for every
cycle $C$ of $N$ with
$R(C)\in \mathcal R'$
all vertices below or equal to $r(C)$ into a vertex and then
labelling that vertex by $R(C)$. Put $t_N:=t|_{V(T_N)}$. 
Then $(T_N,t_N)$ is clearly a 
labelled phylogenetic
tree on $X'$. Since $\mathcal N$ is a labelled level-1 network,
it follows that
$(T_N, t_N)$ is a symbolic discriminating representation of
$\hat{\delta_{\mathcal N}}$. In view of Theorem~\ref{bdtr}, the
proposition follows.
\end{proof}

\begin{algorithm}
\label{avp}
\caption{\textsc{Vertex-Growing} -- Property (P2) is checked in Line 3.}
\SetKw{KwSet}{set}
\KwIn{A symbolic 3-dissimilarity $\delta$ on a set $X$, a 
subset $Y \subseteq X$, and a hierarchy $\mathcal S$ of proper subsets of $Y$.}
\KwOut{A discriminating symbolic representation
  on the partition of $Y$ induced by 
  $\sim_{(Y,\mathcal S)}$ or the statement
  ``There exists no discriminating symbolic
  representation''.}
\BlankLine
Let $Y'$ denote the partition of $Y$ induced by $\sim_{(Y,\mathcal S)}$\;
Apply the \textsc{Bottom-Up} algorithm to the symbolic ultrametric 
$\hat{\delta}$ induced by $\delta$ on $Y'$, as considered
in Proposition~\ref{f3tsu}\;
\If{\textsc{Bottom-Up} returns a labelled tree $\mathcal T$}{
\Return {$\mathcal T$}\;
}
\Else{
\Return{There exists no discriminating symbolic
  representation.} \;
}
\end{algorithm}

To illustrate algorithm \textsc{Vertex-Growing} consider again
the symbolic 3-dissimilarity $\delta_{\mathcal N_1}$ induced by the
labelled level-1 network on $X=\{a\ldots, k\}$ depicted in Fig.~\ref{exnet}(i).
Let $\mathcal M_1$, $\mathcal M_2$, and $\mathcal M_3$ denote the
three labelled simple level-1 networks returned by algorithm
\textsc{Build-Cycle} when given $\delta_{\mathcal N_1}$
such that $L(\mathcal M_1)=X$, $L(\mathcal M_2)=\{b,\ldots,g\}$
and $L(\mathcal M_3)=\{e,f,g\}$. Then the partition of $X$ found in line 1 of
algorithm \textsc{Vertex-Growing} when given $\delta_{\mathcal N_1}$
and $\mathcal R= \bigcup_{i=1}^3\{L(\mathcal M_1)\}$ is 
 $X$ itself, since any two leaves of $X$ are in relation
with respect to $\sim_{(X,\mathcal R)}$. Thus, the discriminating symbolic
representation returned by \textsc{Bottom-Up} is a single leaf.

Armed with the algorithms {\sc Find-Cycles}, {\sc Build-Cycles},
and \textsc{Vertex-Growing}, we next present a pseudo-code version of algorithm 
\textsc{Network-Popping} (Algorithm~\ref{alg:main}).

\begin{algorithm}
\label{abn}
\caption{\textsc{Network-Popping} -- Property (P5) is checked in Line 6.
  \label{alg:main}}
\SetKw{KwSet}{set}
\KwIn{A symbolic 3-dissimilarity $\delta$ on $X$.}
\KwOut{A semi-discriminating  level-1 representation
  $\mathcal N=(N,t')$ of $\delta$, if
  such a representation exists, or the statement 
  ``$\delta$ is not level-1 representable''.}
\BlankLine
Initialize $N$ as an unique vertex $v$, labelled by $X$\;
\KwSet $r=1$\;
Use $\textsc{Find-Cycles}(\delta)$ to obtain $m \geq 0$
pairs $(H_i,R'_i)$ of subsets $H_i$ and $R_i'$ of $X$, $1\leq i\leq m$\;
\If{ for all $i \in \{1, \ldots, m\}$,
  $\textsc{Build-Cycle}(\delta; H_i,R_i)$ returns a labelled
  simple level-1 network $(C_i, t_i)$ as described in that algorithm}{
put $R_i=R(C_i)$, and $\mathcal R=\{R_1, \ldots, R_m\}$\;
\If {for all $i \in \{1, \ldots, m\}$, and all $y, z \in R_i$, and
  $ x \notin R_i$, we have $\delta(x,y)=\delta(x,z)$}{
\While {there exists a leaf $l$ of $N$ whose label set
$V_l\subseteq X$ has two or more elements AND $r \neq 0$}{
\If {there exists $i \in \{1, \ldots, m\}$ such that $V_l=R_i$}{
  identify $l$ with the root of the labelled simple level-1 network
  corresponding to $R_i$ and replace $N$ with the resulting labelled level-1
  network\;
   }
   \Else{
   put $\mathcal S_l=\{R \in \mathcal R : R \subseteq V_l\}$\;
   \If{\textsc{Vertex-Popping}$(\delta, V_l,\mathcal S_l)$ returns 
   a discriminating symbolic representation $\mathcal T=(T,t)$}{
     identify $l$ with the root of $T$ and replace $N$ with the
     resulting labelled level-1
  network\;
   }
   \Else{
   \KwSet $r=0$\;
   }
   }
}
}
}
\If{$r=1$ AND $N$ is not $v$}{
\Return {$\mathcal N:=(N,t')$ where $t'$ is canonically obtained
by combining the maps $t$ and $t_i$, $1\leq i\leq m$}\;
}
\Else{
\Return $\delta$ is not level-1 representable\;
}
\end{algorithm}

To be able to establish in Proposition~\ref{prop:delta-and-delta-N}
that algorithm \textsc{Network-Popping}
returns a semi-discriminating level-1 representation for
a symbolic 3-dissimilarity (if such a representation exists),
we require the following technical result.

\begin{proposition}\label{2imp3}
  Let $\delta$ be a symbolic 3-dissimilarity on $X$ satisfying
  Property~(P1), and 
assume that \textsc{Network-Popping} returns a labelled level-1 network
$\mathcal N$ on $X$ when given $\delta$ as input.
Then the restrictions 
$\delta|_{X\choose \leq 2}$ and $\delta_{\mathcal N}|_{X\choose \leq 2}$
of $\delta$ and $\delta_{\mathcal N}$ to ${X \choose \leq 2}$, respectively, 
coincide if and only if $\delta$ and $\delta_{\mathcal N}$ coincide. 
\end{proposition}

\begin{proof}
Put $\mathcal N=(N,t)$. Also, put $\delta'=\delta|_{X\choose \leq 2}$
and $\delta_{\mathcal N}'=\delta_{\mathcal N}|_{X\choose \leq 2}$.
Clearly, if $\delta$ and $\delta_{\mathcal N}$ coincide then
$\delta'=\delta_{\mathcal N}'$
must hold.

Conversely, assume that 
$\delta'=\delta_{\mathcal N}'$.
Let $Z=\{a,b,c\} \in {X \choose 3}$ and put $m=\delta(Z)$.
Note that since $\mathcal N$ is clearly
a level-1 representation of $\delta_{\mathcal N}$,
Lemma~\ref{gencond} implies  that $\delta_{\mathcal N}$
also satisfies Property (P1). Further note that,
up to permuting the elements in $Z$, we either have 
(i) a $\delta$-fork on $Z$, (ii) $a|bc$ is a $\delta$-triplet,
or (iii) $a||bc$ is a $\delta$-tricycle.

If Case (i) holds then $\delta(a,b)=\delta(a,c)=\delta(b,c)=m$.
Since, by assumption,
$\delta(Y)=\delta_{\mathcal N}(Y)$ for all $Y \in {X \choose 2}$, 
we also have $\delta_{\mathcal N}(a,b)=\delta_{\mathcal N}(a,c)
=\delta_{\mathcal N}(b,c)=m$. Hence, 
$\delta_{\mathcal N}(Z)=m=\delta(Z)$ as $\delta$ satisfies
Property~(P1).

If Case (ii) holds then
$m=\delta(a,b)=\delta(a,c) \neq \delta(b,c)$. Assume for
contradiction that $\delta_{\mathcal N}(Z) \neq m$. Then, since 
$\delta_{\mathcal N}$ satisfies Property (P1) it follows that 
$\delta_{\mathcal N}(Z)=\delta_{\mathcal N}(b,c)$. By Table~\ref{tabtri},
$a||bc$ must be a $\delta_{\mathcal N}$-tricycle. Hence,
there must exist a cycle $C$ in $N$ such that $a \in H(C)$,
$b$ and $c$ are contained in $R(C)$ but
lie on different sides of $C$, and
$t(r(C))=\delta_{\mathcal N}(Z)$. Since
algorithm \textsc{Network-Popping} completes by returning $\mathcal N$
it follows that $C$ is constructed in the while-loop starting in line 16
of algorithm {\sc Build-Cycle}. But then the condition in line 6 of
\textsc{Build-Cycle} has to be satisfied which implies that
$t(r(C))=\delta(Z)$ in view of line 7 of that algorithm. Hence, 
$m\not=\delta_{\mathcal N}(Z)=t(r(C))=\delta(Z)=m$ which is impossible.

If Case (iii) holds then the while-loop initiated in line 16 of
algorithm \textsc{Build-Cycle}
implies that there must exist a cycle $C$ in $N$ such that
$t(r(C))=\delta(Z)=m$.
Since
$\mathcal N$ is returned by algorithm \textsc{Network-Popping}
when given $\delta$ and $\mathcal N$ is clearly a
level-1 representation for $\delta_{\mathcal N}$ it follows that
$\delta_{\mathcal N}(Z)=t(r(C))=m=\delta(Z)$.
\end{proof}

As a first result concerning algorithm \textsc{Network-Popping}, we have

\begin{proposition}\label{prop:delta-and-delta-N}
  Suppose $\delta$ is a symbolic 3-dissimilarity on $X$, and
  \textsc{Network-Popping} applied to $\delta$ returns a labelled
  level-1 network $\mathcal N$. Then $\delta=\delta_{\mathcal N}$. In particular,
  $\mathcal N$ is a level-1 representation for $\delta$.
  \end{proposition}
\begin{proof}
  Put $\mathcal N=(N,t)$.
In view of Proposition~\ref{2imp3}, it suffices to 
show that $\delta(a,b)=\delta_{\mathcal N}(a,b)$ holds
for all $a,b \in X$ distinct.
Let $a$ and $b$ denote two such elements.
We distinguish between the cases that either
 (i) there exists a cycle $C$ of $N$ 
such that $v_{C}(a) \neq v_{C}(b)$, or 
(ii) that no such cycle exists.

Assume first that Case (i) holds. Then $a$ and $b$
lie either on the same side of $C$, or 
one of $a$ and $b$ is below the hybrid $h(C)$ of $C$ and
the other lies on the side of $C$, or $a$ and $b$ lie on different sides of $C$.
If $a$ and $b$ lie on the same side of $C$ or one of them
is below $h(C)$ then we may assume
without loss of generality that there exists a directed 
path in $C$ from $v_{C}(a)$ to $v_{C}(b)$. Then line 22 of
algorithm \textsc{Build-Cycle} implies
$t(v_{C}(a))=\delta(a,b)$. Since 
$lca(a,b)=v_{C}(a)$, it follows that 
$\delta_{\mathcal N}(a,b)=t(v_{C}(a))=\delta(a,b)$, as
required.

If $a$ and $b$ lie on different sides of $C$ then
$x||ab$ is a $\delta$-tricycle, for $x$ as in line 2 of
algorithm \textsc{Build-Cycle}. Since
that algorithm completes, it's line 7 
implies $\delta(a,b)=t(r(C))$.
But then
$\delta_{\mathcal N}(a,b)=t(r(C))=\delta(a,b)$, as $\mathcal N$ is
returned by \textsc{Network-Popping}.

For the remainder, assume that Case (ii) holds, that is,
there exists no cycle $C$ of $N$ 
such that $v_{C}(a) \neq v_{C}(b)$.
Consider the vertex $v_0\in V(N)$ defined as follows: 
if the path from the root $\rho_N$ of
$\mathcal N$ to $lca(a,b)$ does not contain
a vertex that is also contained in a cycle of $N$,
then put $v_0=\rho_N$.
Otherwise let $v_0$ denote the last vertex on a directed path from
$\rho_N$ to $lca(a,b)$ such that $v_0$
belongs to a cycle $Z$ of $N$.
Note that $v_0=lca(a,b)$ holds if $lca(a,b)$
is also contained in $Z$.
Put $V=\mathcal F(v_1)$ where $v_1$ is the unique child of $v_0$
not contained in $Z$, and
let $V'$ denote the partition of $V$
induced by $\sim_{(V,\mathcal S_{v_0})}$
where for any vertex $w\in V(N)$ the set
$\mathcal S_w$ is as defined as in line 12 of algorithm
\textsc{Network-Popping}. 
Let $R_a, R_b\in V'$ such that $a\in R_a$ 
and $b\in R_b$. Then line 5 of 
\textsc{Network-Popping}
implies
$\hat{\delta_{\mathcal N}}(R_a,R_b)=\delta_{\mathcal N}(a,b)$
and $\hat{\delta}(R_a,R_b)=\delta(a,b)$.
Since $\mathcal N$ is returned by 
\textsc{Network-Popping} when given $\delta$,  line 12 of that algorithm
implies 
$\hat{\delta}(R_a,R_b)=\hat{\delta_{\mathcal N}}(R_a,R_b)$.
 Consequently,
$\delta_{\mathcal N}(a,b)=\delta(a,b)$ holds in this case too.
\end{proof}


We conclude this section with remarking that 
the runtime of algorithm \textsc{Network-Popping} is polynomial.
The reasons for this are that  \textsc{Network-Popping}
basically works by comparing $\delta$-trinets 
and $\delta$-values on
subsets on $X$ of size two, and that the number of such
trinets and subsets is polynomial in the size of $X$.

\section{Uniqueness of level-1 representations returned by
  \textsc{Network-Popping}}
\label{sec:uniqueness}

As is easy to see, there exist symbolic 3-dissimilarities that
although they satisfy Properties (P1) - (P6) they are not level-1
representable. The reason for this is that such
3-dissimilarities need not satisfy the assumptions
of lines 10 and 20 in algorithm {\sc Build-Cycle}.
A careful analysis of that algorithm
suggests however two further properties for a symbolic 3-dissimilarity
to be level-1 representable. To state them, we next associate
to a symbolic 3-dissimilarity its CheckLabels graph.

Suppose  $Y_0$, $Y_1$, and  $Y_2$ are three 
pairwise disjoint subsets of $X$ such that for all $x,x' \in Y_0$ and all
$i=1,2$, the graphs 
$TD(Y_i,x)$ and $TD(Y_i,x')$ are isomorphic (which is motivated
by Property (P6)).
Then we denote by $CL(Y_0,Y_1,Y_2)$
the \emph{CheckLabels graph} associated to
$\delta$, $Y_0$, $Y_1$, and $Y_2$ defined as follows. The vertex set of 
$CL(Y_0,Y_1,Y_2)$ is $Y_0 \cup Y_1 \cup Y_2$. Any pair 
$(u,v) \in Y_1 \times Y_2$ is joined by an (undirected) edge $\{u,v\}$, any 
pair $(u,v) \in (Y_1 \cup Y_2) \times Y_0$ is joined by a directed edge 
$(u,v)$, and two elements $u,v \in Y_i$, $i=1,2$, are 
joined by a directed edge $(u,v)$  if there exists a direct path from 
$u$ to $v$ in $TD(Y_i,x)$. Finally, to each 
edge  of $CL(Y_0,Y_1,Y_2)$ with end vertices  $u$ and $v$ or 
directed edge of that graph with tail $u$ and head $v$, we assign the
label $\delta(u,v)$. We illustrate the CheckLabels graph  in
Fig.~\ref{fig:td-graph}(b) for the network $\mathcal N_1$
depicted in Fig.~\ref{exnet}(i).
%


Using the terminology of
algorithm \textsc{Build-Cycle} it is straight-forward to
observe that the following two properties are implied by
\textsc{Build-Cycle}'s lines 10 and 20 whenever its input
  symbolic 3-dissimilarity is level-1 representable:
\begin{enumerate}
\item[\emph{(P7)}] All undirected edges of $CL(H, S_y, S_z)$ 
have the same label;
\item[\emph{(P8)}] For all vertices $u$ of $CL(H, S_y, S_z)$,
  all directed edges in $CL(H, S_y, S_z)$ with tail $u$ have the same label.
 \end{enumerate}

As indicated in Table~\ref{tab:independence},
Properties (P1) - (P8) are independent of each other.
Moreover, they
allow us to characterize level-1 representable symbolic 3-dissimilarities.

\begin{table}[h]
\begin{tabular}{|c|c|c|c|}
\hline
\textsc{Prop.} & $X$ & $M$ & $\delta$ \\
\hline
\hline
\emph{(P1)} & $\{x,y,z\}$ & $\{D,S\}$ &
$\delta(x,y)=\delta(x,z)=\delta(y,z)=D$; \\
&&& $\delta(x,y,z)=S$. \\
\hline
\emph{(P2)} & $\{x,y,z,u\}$ & $\{D,S\}$ &
$\delta(x,y,z)=\delta(y,z,u)=\delta(x,y)=\delta(y,z)=\delta(z,u)=D$;\\
&&& $\delta(Y)=S$ otherwise. \\
\hline
\emph{(P3)} & $\{x_1,x_2,y,z\}$ & $\{D,S_1,S_2\}$ &
$\delta(x_i,y,z)=S_i$, $i \in \{1,2\}$;\\
&&& $\delta(Y)=D$ otherwise. \\
\hline
\emph{(P4)} & $\{x,y,z,u\}$ & $\{D,S\}$ & $\delta(x,y,u)=
\delta(x,u)=\delta(y,z)=\delta(x,y,z)=D$;\\
&&& $\delta(Y)=S$ otherwise. \\
\hline
\emph{(P5)} & $\{1,\ldots,5\}$ & $\{D,S\}$ & $\delta(1,4)=S$; \\
&&& $\delta(Y)=\delta_{\mathcal N_5}(Y)$ otherwise.\\
\hline
\emph{(P6)} & $\{1,\ldots,6\}$ & $\{D,S\}$ & $\delta(3,6)=\delta(2,3,6)=D$;\\
&&& $\delta(Y)=\delta_{\mathcal N_6}(Y)$ otherwise.\\
\hline
\emph{(P7)} & $\{1,\ldots,5\}$ & $\{D,S\}$ &
$\delta(2,4)=\delta(2,3,4)=\delta(1,2,4)=\delta(2,4,5)=S$;\\
&&& $\delta(Y)=\delta_{\mathcal N_7}(Y)$ otherwise.\\
\hline
\emph{(P8)} & $\{1,\ldots,5\}$ & $\{D,S\}$ &
$\delta(3,5)=\delta(3,4,5)=D$;\\
&&& $\delta(Y)=\delta_{\mathcal N_8}(Y)$ otherwise.\\
\hline
\end{tabular}
\caption{\label{tab:independence}
  For sets $X$ and $M$ and $\delta$ a symbolic 3-dissimilarity on $X$
  as indicated, the property stated in the first column of each row
  holds whereas the remaining seven properties do not. For $i=5,6,7,8$,
  the networks $\mathcal N_i$ are depicted in Fig.~\ref{fig:n5678}.
}
\end{table}
\begin{figure}[h]
\begin{center}
\includegraphics[scale=0.6]{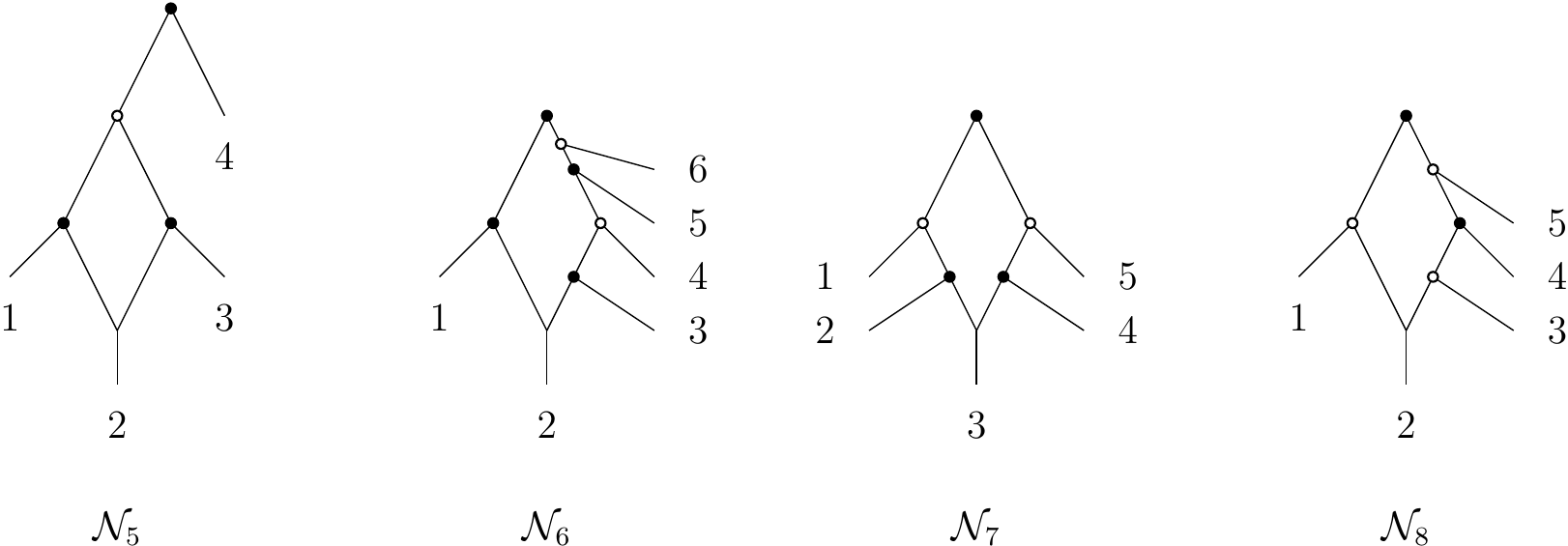}
\caption{\label{fig:n5678} The networks $\mathcal N_i$, $i=5,6,7,8$,
  considered in Table~\ref{tab:independence}}
\end{center}
\end{figure}

\begin{theorem}\label{theo:char}
Let $\delta$ be a symbolic 3-dissimilarity on $X$. Then the 
following statements are equivalent (where in (iii)-(v) the input to
algorithm \textsc{Network-Popping} is $\delta$):\\
(i) $\delta$ is level-1 representable.\\
(ii) $\delta$ satisfies conditions (P1) - (P8).\\
(iii) \textsc{Network-Popping} returns a
labelled level-1 network  which is unique up to isomorphism.\\
(iv) \textsc{Network-Popping} returns a 
level-1 representation for $\delta$.\\
(v) \textsc{Network-Popping} returns a 
semi-discriminating level-1 representation for
$\delta$.
\end{theorem}

\begin{proof}
  (i) $\Rightarrow$ (ii): This is an immediate consequence of
  Lemma~\ref{gencond}, Proposition~\ref{cocy}, the remark 
  preceding Proposition~\ref{cytdx} and the observation preceding
  Table~\ref{tab:independence}.

  (ii) $\Rightarrow$ (iii): Assume that $\delta$  satisfies
  Properties (P1) - (P8). Then
algorithm \textsc{Find-Cycles} first constructs the graph
$\mathcal C(\delta)$ and then finds for each connected component $K$ of
$\mathcal C(\delta)$ the pair $(H_K,R_K')$. Since algorithm
\textsc{Build-Cycles} relies on Properties (P3), (P4), (P6) - (P8)
being satisfied, it follows that \textsc{Build-Cycles} constructs for
each pair $(H_K,R_K')$, $K$ a connected component of $\mathcal C(\delta)$, a
labelled simple level-1 network as specified in the
output of \textsc{Build-Cycles}. By construction, the labelled DAG
$\mathcal N=(N,t)$
returned by algorithm \textsc{Network-Popping} is clearly a labelled
phylogenetic network. 
Since, in view of the while loop of that
algorithm starting at line 7, no two cycles in $N$ can share a vertex
it follows that $N$ is in fact a level-1 network. Proposition~\ref{2imp3}
combined with the observation that
in none of our four algorithms we have to break a tie
implies  that $\mathcal N$ is
unique up to isomorphism.

(iii) $\Rightarrow$ (iv):
This is trivial in view  of  Proposition~\ref{prop:delta-and-delta-N}.

(iv) $\Rightarrow$ (v): Suppose  algorithm \textsc{Network-Popping}
returns a 
level-1 representation $\mathcal N$ for $\delta$.
To see that $\mathcal N$ is in fact semi-discriminating,
note that algorithms \textsc{Vertex-Growing} and 
\textsc{Build-Cycles} return a
discriminating symbolic representation and a
discriminating level-1 representation for its input symbolic 3-dissimilarity,
respectively. In combination
it follows that $\mathcal N$ must be semi-discriminating.

(v) $\Rightarrow$ (i): This is trivial.
\end{proof}

As suggested by the two semi-discriminating level-1 representations
$\mathcal N_1$ and $\mathcal N_3$
for $\delta_{\mathcal N_1}$ depicted in Fig.~\ref{exnet}(i) and (ii),
the output of  algorithm \textsc{Network Popping}
when given a level-1 representable symbolic 3-dissimilarity $\delta$
need not be the labelled level-1 network that
induced $\delta$. To help clarify the relationship between both networks,
we require further terminology.

Suppose that $(N,t)$ is a labelled level-1 network.
Then we say that a cycle $C$ of $N$ is 
\emph{weakly labelled} if there exists at least one vertex $v$ 
on either side of $C$ such that $t(v) \neq t(r(\mathcal C))$.
More generally, we call a labelled level-1 network $(N,t)$
{\em weakly labelled} if every cycle of $N$ is weakly labelled. 
For example,
the labelled level-1 network $\mathcal N_2$ pictured in
Fig.~\ref{exnet}(ii) 
is weakly labelled (but not semi-discriminating) whereas the
network $\mathcal N_3$ depicted in Fig.~\ref{exnet}(iii) 
is semi-discriminating but not weakly labelled. 

Armed with this definition, we can characterize
weakly labelled cycles as follows.

\begin{lemma}\label{ptric}
Let $\mathcal N=(N,t)$ be a labelled level-1 network, 
and let $C$ be a cycle of $N$. Then $C$ is weakly 
labelled if and only if there exists some 
$x\in H(C)$ and leaves $y,z\in R(C)-H(C)$ that lie
 on different sides of $C$ such that  $x|| yz$ is a 
$\delta_{\mathcal N}$-tricycle. Moreover, $x' || yz$ is a
 $\delta_{\mathcal N}$- tricycle, for all $x' \in H(C)$.
\end{lemma}

\begin{proof}
Put $\delta=\delta_{\mathcal N}$.
Assume first that there exists some $x \in H(C)$
and leaves $y,z\in R(C)-H(C)$ that lie on two different sides 
of $C$ such that $x||yz$ is a $\delta$-tricycle. Then 
$\delta(x,y,z)=\delta(z,y)=t(r(\mathcal C))$. Also 
$\delta(x,y)=t(v_{C}(y))$ and $\delta(x,z)=t(v_{C}(z))$. 
In view of Table~\ref{tabtri},
$\delta(x,y,z)\not\in \{\delta(x,y),\delta(x,z)\}$
and, so, $t(v_{C}(i)) \neq t(r(C))$, for $i=y,z$. 

Conversely, suppose $C$ is weakly labelled. 
Let $v_1,v_2\in V(C)$ denote two vertices of $N$ that lie on 
different directed paths from $r(C)$ to $h(C)$ such that
$t(r(C))\not\in\{t(v_1),t(v_2)\}$. Suppose $y,z\in X $ are such that
$v_C(y)=v_1$ and $v_C(z)=v_2$. Then
$x||yz$ must be a $\delta$-tricycle, for all $x \in H(C)$.
Indeed, $\delta(x,y)=t(v_1)$ and $\delta(x,z)=t(v_2)$ holds. Since 
$\delta(y,z)=\delta(x,y,z)=t(r(C)) \notin\{\delta(x,y),\delta(x,z)\}$, 
Table~\ref{tabtri} implies that 
$x||yz$ is a $\delta$-tricycle.

The remainder of the lemma follows from the fact that,
for all $x' \in H(C)$, we have $\delta(x',y,z)=\delta(x,y,z)$, 
$\delta(x,y)=\delta(x',y)$ and $\delta(x,z)=\delta(x',z)$.
\end{proof}

As a consequence, we can strengthen Proposition~\ref{cocy} to the
following characterization.

\begin{theorem}\label{cocy2}
If $\mathcal N=(N,t)$ is a labelled level-1 network, the connected 
components of $\mathcal C(\delta_{\mathcal N})$ are in 1-1 correspondence 
with the weakly labelled cycles of $N$.
\end{theorem}

Implied by Theorem~\ref{cocy2}, we have

\begin{corollary}\label{corcor}
Let $\delta$ be a level-1 representable symbolic 3-dissimilarity on $X$, 
and let $\mathcal N=(N,t)$ be the level-1 representation of $\delta$
returned by algorithm \textsc{Network-Popping} when applied to $\delta$.
Then $\mathcal N$ is weakly labelled if and only if, 
for any level-1 representation $\mathcal N'=(N',t')$ of 
$\delta$, the number of cycles in $N$ equals the number of 
weakly labelled cycles in $N'$.
In particular, 
the number of cycles in $N$ is minimal.
\end{corollary}

\begin{algorithm}
\label{asdtc}
\caption{\textsc{Transform}}
\SetKw{KwSet}{set}
\KwIn{A labelled level-1 network $\mathcal N=(N,t)$ on $X$.}
\KwOut{A semi-discriminating, weakly labelled, partially resolved
  level-1 network
  $\mathcal N'=(N',t')$ such that $\delta_{\mathcal N}=\delta_{\mathcal N'}$.}
\BlankLine
\KwSet{$\mathcal N'=\mathcal N$}\;
\While{$\mathcal N'$ is not semi-discriminating or not weakly labelled
or not partially resolved}{
Collapse all edges $(u,v)$ satisfying $t'(u)=t'(v)$ and such that either $u$ and $v$ belong to the same cycle of $N'$ or do not belong to a cycle\;
\For{All vertices $v$ of a cycle $C$ of degree 4 or more}{
Define a new child $w$ of $v$\;
\KwSet{$t'(w)=t'(v)$}\;
\If{$v=r(C)$}{
Redefine the children of $v$ in $C$ as children of $w$\;
}
\Else{
Redefine the children of $v$ outside of $C$ as children of $w$\;
}
}
\For{All cycles $C$ of $N'$ such that $(r(C),h(C))$ is an edge of $N'$}{
Remove the edge $(r(C),h(C))$\;
}
Remove all vertices of degree 2\;
}
\end{algorithm}

\begin{corollary}
  \label{cor:uniqueness}
Suppose $\mathcal N$ is a labelled level-1 network 
and $\mathcal N'$ is the level-1 representation for
$\delta_{\mathcal N}$ returned by algorithm \textsc{Network-Popping}.
Then $\mathcal N'$ is isomorphic with the labelled level-1 network
returned by algorithm
\textsc{Transform} when given $\mathcal N$ as input. In particular,
$\mathcal N$ and $\mathcal N'$ are isomorphic
if and only if $\mathcal N$ is semi-discriminating,
weakly labelled, and partially resolved.
Furthermore, if $\delta$ is a level-1 representable symbolic 3-
dissimilarity, then there exists an unique representation of $\delta$ that is 
semi-discriminating, weakly labelled, and partially resolved.
\end{corollary}

\section{Characterizing level-1 representable symbolic 3-dissimilarities}
\label{sec:characterizing-level-1}
In this section, we present a characterization of
level-1 representable symbolic
3-dissimilarities on $X$ in terms of level-1 representable symbolic
3-dissimilarities on subsets of $X$ of size $|X|-1$
(Theorem~\ref{iffsubsets}). Combined with the fact that
algorithm {\sc Network-Popping} has polynomial run time, this suggests
that {\sc Network-Popping} might lend itself to studies involving
large data sets using a Divide-and-Conquer approach.

At the heart of the proof of our characterization
lies the following technical lemma which
concerns the question under what circumstances the restriction of a
level-1 representable symbolic 3-dissimilarity $\delta$ on $X$ is itself
level-1 representable. Central to its proof is the fact that
$|X|\not =4$ since, in general, a symbolic 3-dissimilarity $\delta$
on a set $X$ of size $4$ need not be level-1 representable
but the restriction of $\delta$ to any subset of size 3 is
level-1 representable. An example for this is furnished
by the symbolic 3-dissimilarity $\delta$ on $X=\{x,y,z,u\}$,
given by $\delta(x,y,z)=\delta(y,z,u)=\delta(x,y)=\delta(y,z)=\delta(z,u)
\neq \delta(x,z)=\delta(x,u)=\delta(y,u)=\delta(x,z,u)=\delta(x,y,u)$. 

Using the assumptions and definitions for the elements
  $x$, $y$, and $z$, and the sets $H$, $S_z$, and $S_y$
  made in algorithm {\sc Build-Cycle}, we have the following
 result.

\begin{lemma}\label{shortpath}
  Suppose $\delta$ is a symbolic 3-dissimilarity on $X$
  satisfying Properties (P1), (P2), (P4), and (P6), $x||yz$ is the
  $\delta$-tricycle chosen in line 2 of algorithm {\sc Build-Cycle},
  and $i\in \{y,z\}$. If $u,w \in S_i$
  are joined by a direct path from $u$ to $w$
  in $TD(S_i,x)$, then either $(u,w)$ is a directed
  edge of $TD(S_i,x)$ or there exists $v \in S_i$ such that both
  directed edges $(u,v)$ and $(v,w)$ are contained in $TD(S_i,x)$.
\end{lemma}

\begin{proof}
By symmetry, we may assume $i=y$.
Suppose there
exists a directed path $v_0=u, v_1, \ldots, v_k, v_{k+1}=w$,
some $k\geq 0$, from $u$ to $w$ in $TD(S_y,x)$
and that $(u,w)$ is not a directed
edge on that path. Then $k\geq 1$ and, so, $v_1\not\in\{u,w\}$.
It suffices to show that $(v_1,w)$
is a directed edge of $TD(S_y,x)$.

Observe first that, in view of Property (P6), $(w,u)$ is not a
directed  edge in $TD(S_y,x)$ as otherwise $TD(S_y,x)$ would contain a
directed cycle. Combined with the definition of $S_y$ it follows
that either $x|uw$ is a $\delta$-triplet or we have
a $\delta$-fork on $\{x,u,w\}$. In either case, $\delta(u,x)=\delta(w,x)$
holds. Since $(u,v_1)$ is a directed edge in $TD(S_y,x)$, we also have
that $xv_1|u$ is a $\delta$-triplet.
Hence, $\delta(v_1,x)\neq\delta(x,u)=\delta(w,x)$ and so we cannot have
a $\delta$-fork on $\{x,w,v_1\}$. Since, in view of Property (P4),
we cannot have a 
$\delta$-tricycle on $\{x,w,v_1\}$ either 
$\delta(w,v_1)=\delta(w,x)$ or $\delta(w,v_1)=\delta(v_1,x)$ follows.

If the first equality holds, then $v_1x|w$ is a 
$\delta$-triplet and, so, $(w,v_1)$ is a directed edge in $TD(S_y,x)$.
Consequently, the directed path $v_1, \ldots, v_k, w$ concatenated with
that edge forms a directed cycle in $TD(S_y,x)$, which is impossible in view
of Property (P6) holding.
Thus, $\delta(w,v_1)=\delta(v_1,x)$ must hold. Consequently, $wx|v_1$
is a $\delta$-triplet and, so, $(v_1,w)$ is an edge in $TD(S_y,x)$,
as required.
\end{proof}

To establish the main result of this section (Theorem~\ref{iffsubsets}), 
we need to be able to distinguish between the sets defined in lines~8 and 9 of
algorithm {\sc Build-Cycle} when given a symbolic 3-dissimilarity $\delta$
on $X$ and the restriction $\delta|_Y$ of $\delta$ to a
subset $Y\subseteq X$ with $|Y|\geq 3$.
To this end, we augment for a symbolic 3-dissimilarity $\kappa$ on $X$
the definition of those sets by writing $S_i(\kappa)$
rather than $S_i$, $i=y,z$.

Observe first that if $\delta$ is level-1 representable and $Y\subseteq X$
such that $|Y|\geq 3$, then the restriction $\delta|_Y$
of $\delta$ to $Y$ is clearly level-1 representable. Indeed, a
level-1 representation $\mathcal N(\delta|_Y)$ of $\delta|_Y$
can be obtained from a level-1 representation $\mathcal N(\delta)$
of $\delta$ using the following 2-step process.
First, remove all leaves in $X-Y$ and 
their respective incoming edges from $\mathcal N(\delta)$ and then suppress
all resulting degree two vertices. Next, 
apply algorithm \textsc{Transform} to the resulting network. This begs
the question of when level-1 representations of symbolic 3-dissimilarities
on subsets of $X$ give rise to a level-1 representation of a symbolic
3-dissimilarity on $X$. To answer this question which is the purpose of
Theorem~\ref{iffsubsets} we require the next result.

\begin{proposition} \label{lem:n=4,5}
  Let $\delta$ be a symbolic 3-dissimilarity on $X$. Then the following
  statements hold.
  \begin{enumerate}
  \item[(i)] If $|X|\geq 6$ and $\delta$ does not satisfy Property (Pi),
    $i\in\{1,2,\ldots, 8\}$,
    then there exists some
    $Y\subseteq X$ with $3\leq |Y|\leq 5$ such that that property is also not
    satisfied by $\delta|_Y$.
    \item[(ii)] If $|X|\geq 6$ and $\delta$ is not level-1 representable then
      there exists some $Y\subseteq X$ with $3\leq |Y|\leq 5$ such
      that $\delta|_Y$
    is also not level-1 representable.
\end{enumerate}
\end{proposition}

\begin{proof}[Proof]
  (i) The proposition is straight-forward to show for Properties
  (P1) and (P2), since  they involve three
  and four elements of $X$, respectively. Note that to
  see  Property (P$i$), $3\leq i\leq 8$,
  we may assume without loss of generality that Properties (P$j$),
  $3\leq j\leq i-1$, are satisfied by $\delta$. For ease of readability,
  we put $S_y:=S_y(\delta)$.

  If $\delta$ does not satisfy Property (P3) then there exists a connected
  component $C$ of $\mathcal C(\delta)$ and $\delta$-tricycles $\tau,
  \tau' \in V(C)$ such that
  $\delta(L(\tau))\not=\delta(L(\tau'))$. Without loss of generality,
  we may assume that $\tau$ and $\tau'$ are adjacent. Then
  $|L(\tau)\cap L(\tau')|=2$. Let $x,y,z\in X$
  such that  $\tau=x||yz$. Then either $\tau'=x'||yz$ or $\tau'=x||yz'$
  where $x',z'\in X$. But then Property (P3) is not
  satisfied either for $\delta$ restricted to
  the 5-set $Z=\{x,y,z,x',z'\}$.

  For the remainder, let $(H,R')$ denote the pair
  returned by algorithm \textsc{Find-Cycles} when given $\delta$
  and let $x\in H$
  and $y,z\in R'$ such that $x||yz$ is a vertex in the connected
  component $C$ of $\mathcal C(\delta)$ corresponding to $(H,R')$.
  Suppose $\delta$ does not satisfy Property (P4). Assume first that
  the second part of Property (P4) is not satisfied.
  Then if there exists any element
  $u$ contained in $ H\cap S_y$  or in $ H\cap S_z$  or in
  $ S_z\cap S_y$ then $u$ is also contained in the corresponding intersections
  involving the sets  $S_y(\delta|_Z)\subseteq S_y$ and
  $S_z(\delta|_Z) \subseteq S_z$ 
  found by \textsc{Build-Cycle} in its lines lines 8 and 9
  for $\delta$ restricted to  $Z=\{ x,y,z,u\}$.
  Thus, the second part of Property (P4) does not hold for $\delta|_Z$.

  Now assume that the first part of Property (P4) does not hold
  for $\delta$, that is,
$S_i'\not =A:=\{w\in S_i : \delta(w,x)\not=\delta(y,z)\}$.
By symmetry, we may assume without loss of generality that $i=y$. 
Then since $S_y'\subseteq A$ clearly
holds there must exists some $w\in A-S_y'$. Put $U=\{x,y,z,w\}$.
Then $w\not\in S_y'(\delta|_U)$ as $w\not\in S_y'$. However
we clearly have that $w\in S_y(\delta|_U)$ and
$\delta|_U(w,x)\not=\delta|_U(y,z)$.
Thus, the first part of Property (P4) is not satisfied with
$\delta$ replaced by $\delta|_U$.

  If $\delta$ does not satisfy Property (P5) then since
   $y\in R:=H\cup S_y\cup S_z$
  it follows for $u:=y$ and $v$ and $w$ as in the statement of
  Property (P5) that the restriction of $\delta$ to 
  $\{x,u,z,v,w\}$ does not satisfy Property (P5) either.

  If $\delta$ does not satisfy Property (P6)
   then either (a)
   there exist elements $u,u'\in H$ such that $TD(S_y, u)$
   and $TD(S_y, u')$ are not isomorphic or (b) there exists some
   $u\in H$ such that $TD(S_y, u)$ has a directed cycle $C$.
   
   Assume first that Case (a) holds. Then there must exist
   distinct vertices $v$
   and $w$ in $S_y$ such that $(v,w)$ is a directed edge in
   $TD(S_y, u)$ but not in $TD(S_y, u')$. With $Z=\{v,u,u',w,z\}$
   it follows that $S_v(\delta|_Z)=\{v,w\}$. Since  the directed
   edge $(v,w)$ is clearly contained in the TopDown graph
   $TD(\{v,w\},u)$ associated to $\delta|_Z$ but not in the
   TopDown graph $TD(\{v,w\},u')$ associated to $\delta|_Z$,
   Property (P6) is not satisfied for $\delta|_Z$.

Thus, Case (b) must hold. In view of
   Proposition~\ref{cytdx}(i), we may assume that the size of
   $C$ is three. Hence, the subgraph $G$ of $TD(S_y,u)$ induced by
   $Z=V(C)\cup \{z,u\}$ also contains a cycle of length 3. Since
   $G$ coincides with the TopDown graph $TD(V(C),u)$ for $\delta|_Z$
   and $|Z|=5$ holds, it follows that $\delta|_Z$ does not satisfy
   Property (P6).

  If $\delta$ does not satisfy Property (P7) then there must exist
  undirected edges $e=\{a,b\}$ and $e'=\{a',b'\}$ in
  $CL(H, S_y,S_z)$ such that $\delta(a,b)\not=\delta(a',b')$.
  Then for at least one of $e$ and $e'$, say $e$, we must have
  that $\delta(a,b)\not=\delta(y,z)$. Put $Z=\{x,y,z,a,b\}$.
  Then since $\{y,z\}$ is also
  an undirected edge in $CL(H, S_y(\delta|_Z),S_z(\delta|_Z))$ it follows that
  $\delta|_Z$
   does not satisfy Property (P7) either.

   Finally, suppose that $\delta$ does not satisfy Property (P8).
   Considering both alternatives in the statement of Property (P8) together,
there must exist vertices $u\in S_y$ and $v,w\in S_y \cup H$
such that both $(u,v)$ and $(u,w)$ are directed edges of $CL(H, S_y,S_z)$ and
  $\delta(u,v)\not=\delta(u,w)$. Independent of whether
  $v,w\in S_y$ or $v,w\in H$ or $v\in S_y$ and $w\in H$,
it follows that either $\delta(u,x)\neq\delta(u,v)$ or
$\delta(u,x) \neq \delta(u,w)$.  Assume without loss of generality
that $\delta(u,x)\not=\delta(u,v)$. Note that $(u,x)$ is also a directed
edge in $CL(H, S_y,S_z)$.

If $v \in H$, then $\delta|_Z$ does not satisfy Property
(P8) for $Z=\{x,y,z,u,v\}$.
So assume $v\not\in H$. Then $v \in S_y$. Since $(u,v)$ is a directed
edge in $CL(H, S_y,S_z)$ it follows that there exists
a directed path $P$ from $u$ to $v$ in $TD(S_y,x)$. By
Lemma~\ref{shortpath}, either (a) $P$ has a single directed edge
or (b) there exists some $v_1 \in S_y$ such that both
$(u,v_1)$ and $(v_1,v)$ are directed edges of $TD(S_y,x)$.
  
If Case (a) holds, then $\delta|_Z$ does not satisfy Property (P8)
for $Z=\{x,y,z,u,v\}$. So assume that Case (b) holds. Then
$\delta|_{Z'}$ does not satisfy Property (P8) for $Z'=\{x,y,z,u,v,v_1\}$. 
Since the definition of $TD(S_y,x)$ implies that  $xv|v_1$ is a
$\delta$-triplet, it follows that $\delta(x,v) \neq \delta(x,v_1)$.
Hence, either $\delta(v,x) \neq \delta(v,z)$ or
$\delta(v_1,x) \neq \delta(v,z)$. By Properties (P3) and (P4)
it follows  in the first case that $x||vz$ is a $\delta$-tricycle,
and that $x||v_1z$ is a $\delta$-tricycle in the second case.
Thus, either $v$ or $v_1$ can play the role of $y$ in $\tau$. Consequently,
$\delta$ restricted to $Z=Z'-\{y\}$ does not satisfy Property (P8).

  (ii) This is a straight-forward consequence of
  Theorem~\ref{theo:char} and Proposition~\ref{lem:n=4,5}(i).
\end{proof}

\begin{theorem}\label{iffsubsets}
  Let $\delta$ be a symbolic 3-dissimilarity on a set $X$ such that
  $|X| \geq 6$. Then $\delta$ is level-1 representable if and only if for
all subsets $Y \subseteq X$ of size $|X|-1$, the restriction
  $\delta|_Y$ is level-1
  representable.
\end{theorem}

\begin{proof}
  Suppose first that $\delta$ is level-1 representable. Then, by
  the observation preceding
  Proposition~\ref{lem:n=4,5}, $\delta|_Y$ is level-1 representable, for all
  subsets $Y\subseteq X$ of size $|X|-1$.

  Conversely, suppose that $X$ is such that for all subsets $Y \subseteq X$ of
  size $|X|-1$, the restriction
  $\delta|_Y$ is level-1 representable but that $\delta$ is
  not level-1 representable. Then, by Proposition~\ref{lem:n=4,5}
  there exists a
  subset $Y\subseteq X$ with $|Y|\in \{3,4,5\}$ such that 
  $\delta|_Y$ is also not level-1 representable. But then
  $\delta$ restricted to any
  subset $Z$ of $X$ size $|X|-1$ that contains $Y$ 
    also is not level-1 representable which is impossible.
  \end{proof}

\section{Conclusion}
\label{sec:conclusion}

In this paper, we have introduced the novel {\sc Network-Popping}
algorithm. It takes as input an orthology relation,
formalized as a symbolic 3-dissimilarity $\delta$, and finds, 
 in polynomial time,
a level-1 representation of $\delta$
precisely if such a representation exists. In addition to 
this representation being a
discriminating symbolic representation of $\delta$ precisely if such a
tree is supported by $\delta$, {\sc Network-Popping} enjoys
several other attractive properties. As part of our analysis of
{\sc Network-Popping},
we characterize level-1 representable symbolic 3-dissimilarities
$\delta$ in terms of eight natural
properties that $\delta$ must satisfy. Last-but-not-least,
we also characterize a level-1 representable symbolic 3-dissimilarity
$\delta$ on some set $X$ with $|X|\geq 6$
in terms of level-1 representable orthology
relations induced by  $\delta$ on subsets of $X$  of size
$|X|-1$. Combined with the
polynomial run-time of {\sc Network-Popping} this suggests that
it could potentially be applied to large data sets
within a Divide-and-Conquer framework thus 
providing an alternative to tree-based reconciliation or
error correction approaches for orthology relations.

However a number of open questions remain. For example can other
types of phylogenetic networks be used to also represent orthology
relations. Interesting types of such networks might be tree-child networks
\cite{vIM14} as they are uniquely determined by the trinets they induce and
also regular networks \cite{W11} as they are known to be uniquely determined
by the phylogenetic trees they induce, a property that is not shared by
phylogenetic networks in general \cite{GH12}. For those  networks it would
also be interesting to understand how the representation
of an orthology relation
in terms of those trees relates to the way such a relation is represented
by the labelled network displaying the trees. Motivated by the point made in
\cite[Chapter 12]{F03} on $k$-estimates, $k\geq 3$, 
already mentioned above it might also be
interesting to investigate if symbolic $k$-dissimilarities for $k\geq 4$
lend themselves as a useful formalization of orthology relations.

A further question concerns the fact that by evoking parsimony
we only distinguish between 3 types of trinets associated to an
orthology relation. Thus it might be interesting to investigate what
can be done if this framework is replaced by e.\,g.\,a probabilistic one
which assigns probability values to the trinets.

\section*{Acknowledgements}
The authors would like to thank M. Taylor for stimulating discussions
on orthology relations.
  
\bibliographystyle{abbrv}
\bibliography{bibliography}{}

\end{document}